\documentclass[pdflatex,sn-mathphys-num]{sn-jnl}

\usepackage{graphicx}%
\usepackage{multirow}%
\usepackage{amsmath,amssymb,amsfonts}%
\usepackage{amsthm}%
\usepackage{mathrsfs}%
\usepackage[title]{appendix}%
\usepackage{xcolor}%
\usepackage{textcomp}%
\usepackage{manyfoot}%
\usepackage{booktabs}%
\usepackage{algorithm}%
\usepackage{algorithmicx}%
\usepackage{algpseudocode}%
\usepackage{listings}%
\usepackage{enumitem}

\theoremstyle{thmstyleone}%
\newtheorem{theorem}{Theorem}
\newtheorem{proposition}[theorem]{Proposition}%

\theoremstyle{thmstyletwo}%
\newtheorem{example}{Example}%
\newtheorem{remark}{Remark}%

\theoremstyle{thmstylethree}%
\newtheorem{definition}{Definition}%
\newtheorem{lemma}{Lemma}
\newtheorem{corollary}{Corollary}
\usepackage{array}
\usepackage{soul}
\numberwithin{equation}{section}
\begin{document}

\title[Article Title]{Construction of Self-Orthogonal Quasi-Cyclic Codes and Their Application to Quantum Error-Correcting Codes}


\author[1]{\fnm{Mengying} \sur{Gao}}\email{linxiing2022@163.com}

\author*[1]{\fnm{Yuhua} \sur{Sun}}\email{sunyuhua\_1@163.com}

\author[1]{\fnm{Tongjiang} \sur{Yan}}\email{yantoji@163.com}

\author[1]{\fnm{Chun'e} \sur{Zhao}}\email{zhaochune1981@163.com}

\affil[1]{\orgdiv{Colleges of Sciences}, \orgname{China University of Petroleum}, \orgaddress{\city{Qingdao}, \postcode{266580}, \state{Shandong}, \country{China}}}




\abstract{In this paper, necessary and sufficient conditions for the self-orthogonality of $t$-generator quasi-cyclic (QC) codes are presented under the Euclidean, Hermitian, and symplectic inner products, respectively. 
 Particularly, by studying the structure of the dual codes of a class of 2-generator QC codes, we derive necessary and sufficient conditions for the QC codes to be dual-containing under the above three inner products.
 This class of 2-generator QC codes generalizes many known codes in the literature. Based on the above conditions, we construct several quantum stabilizer codes and quantum synchronizable codes with good parameters, some of which share parameters with certain best-known codes listed in 
 the codetable of Grassl.}




\maketitle
\section{Introduction}\label{sec1}

To ensure the reliability of classical communication, introducing error-correcting codes is an effective way to safeguard data integrity against noise and interference during transmission. Cyclic codes, an important subclass of error-correcting codes, possess efficient encoding and decoding algorithms due to their well-structured algebraic properties. As a natural generalization of cyclic codes, quasi-cyclic (QC) codes not only preserve the structural advantages of cyclic codes but also provide broader applicability and increase flexibility in parameter selection. Therefore, QC codes have attracted extensive attention and in-depth study  and numerous record-breaking classical codes have been derived from QC codes~\cite{Kasami1974,Siap2000,Chen2018,Xie2024,Meng2025}.

Compared to classical communication systems, quantum communication systems are more susceptible to external disturbances. Therefore, quantum error-correcting codes (QECCs) are crucial for protecting quantum information against quantum errors, including bit-flip, phase-flip, and synchronization errors. In the 1990s, Shor et al. constructed the first QECC, marking the beginning of this field~\cite{Shor1995}.
Subsequently, Calderbank et al. showed that binary quantum codes can be systematically derived from self‑orthogonal classical codes, thereby recasting the construction of quantum codes as the design of classical linear codes with self‑orthogonal or dual‑containing properties~\cite{Calderbank1998}.

Since then, research on QECCs has progressed rapidly, and numerous codes with good parameters have been constructed by leveraging self-orthogonal or dual-containing classical codes~\cite{Ashikhmin2001,Ketkar2006,Pang2020,Li2025}, most of which belong to the class of quantum stabilizer codes used to correct bit-flip and phase-flip errors. In 2013, quantum synchronizable codes (QSCs) were first proposed by Fujiwara as another subclass of QECCs~\cite{Fujiwara2013,Fujiwara2013-2}. These codes are capable of correcting not only bit-flip and phase-flip errors but also synchronization errors. Following that, a variety of quantum synchronizable codes with desirable properties have been developed~\cite{Fujiwara2014,Luo2018,Li2019,Luo2019,Dinh4ps,Liu2021,Sun2023,Liu2023,Liu2024QIP}. It should be emphasized that nearly all of the aforementioned QECCs are derived from classical cyclic, constacyclic, Bose–Chaudhuri–Hocquenghem (BCH) codes or algebraic geometry codes.

Owing to the favorable algebraic properties of QC codes, the construction of QECCs based on QC codes has recently emerged as a prominent research focus~\cite{Ling2001}. In 2018, Galindo et al. studied the dual structure of a class of 2-generator QC codes under the Euclidean, symplectic, and Hermitian inner products, and established sufficient conditions for their dual-containing property. Based on these codes, they constructed quantum stabilizer codes with good parameters~\cite{Galindo2018}. Since then, significant progress has been made in the study of 1-generator and 2-generator QC codes, especially regarding their self-orthogonality and dual-containing properties under symplectic and Hermitian inner products, resulting in the construction of many high-performance quantum stabilizer codes~\cite{Lv2020,Lv2020-2,Lv2021CAM,Guan2022-2,Guan2022-1,Guan2024,BS2024}. 
In addition, in order to obtain quantum stabilizer codes with better performance, Ezerman et al. also established the necessary and sufficient conditions for the self-orthogonality of certain quasi-twisted codes from the perspective of direct sum decomposition~\cite{Grassl2025}.

To investigate the dual-containing properties of QC codes, it is essential to conduct a thorough analysis of the structure of their dual codes. However, the study of the dual code structure of QC codes remains challenging and scarce. As far as we know, Abdukhalikov et al. investigated the dual codes of 1-generator QC codes of index 2~\cite{KA2023}. Then Benjwal and Bhaintwal extended this result to QC codes with an arbitrary index $l$, and employed them to construct quantum stabilizer codes~\cite{BS2024}. 
Moreover, the application of QC codes to the construction of QSCs was first realized by Du et al. in 2023, introducing new insights and methods for further development in this area~\cite{duchao2023}.

Inspired by the above work, we investigate the self-orthogonality of $t$-generator QC codes with respect to the Euclidean, Hermitian and symplectic inner products, and provide the corresponding necessary and sufficient conditions for these codes to be self-orthogonal. Building upon the above results, we further investigate the algebraic structure of the duals of a class of 2-generator QC codes and derive the necessary and sufficient conditions for them to be self-orthogonal and dual-containing under the three aforementioned inner products. As an application of the above analysis, we finally construct several types of quantum stabilizer codes and quantum synchronizable codes.

The remainder of this paper is organized as follows. Section \ref{sec2} introduces the basic knowledge of QC codes. Necessary and sufficient conditions for $t$-generator QC codes to be self-orthogonal under three inner products are given in Section \ref{sec3}. The structure of the dual codes of a class of 2-generator QC codes is investigated in Section \ref{sec4}. Three classes quantum stabilizer codes are constructed in Section~\ref{sec5}. 
Based on Euclidean dual-containing 2-generator QC codes, a class of QSCs is obtained in Section \ref{sec6}. Finally, we summarize the main results and discuss the potential directions for future research in Section~\ref{sec7}. 

\section{Preliminaries}\label{sec2}
Let $p$ be a prime and $q = p^{\gamma}$, where  $\gamma$  is a positive integer. The finite field with  $q$  elements is denoted by  $\mathbb{F}_q$. For positive integers $n$, $k$, an $[n,k]$ linear code over \( \mathbb{F}_q \) is a $k$-dimension subspace of \( \mathbb{F}_q^n \), where $n$ is called the length of the code and $k$ is called its dimension. 

\begin{definition}
    Let $m$, $l$ be positive integers and  $\mathcal{C}$ a linear code of length $ml$ over $\mathbb{F}_q$. If for any codeword  
    \[
    \mathbf{c} = (c_{0,0}, \ldots, c_{l-1,0}, c_{0,1}, \ldots, c_{l-1,1}, \ldots, c_{0,m-1}, \ldots, c_{l-1,m-1}) \in \mathcal{C},
    \]  
    its $l$ right cyclic shift  
    \[
    (c_{0,m-1}, \ldots, c_{l-1,m-1}, c_{0,0}, \ldots, c_{l-1,0}, \ldots, c_{0,m-2}, \ldots, c_{l-1,m-2}) \in \mathcal{C},
    \]  
    then $\mathcal{C}$ is called a quasi-cyclic (QC) code of length $ml$ and index $l$. If the linear combinations of all the row vectors of a matrix $M$ exactly span all the codewords of $\mathcal{C}$, then $M$ is called a generator matrix of $\mathcal{C}$. In particular, $\mathcal{C}$ is a cyclic code  when $l = 1$.
\end{definition}
Cyclic codes can be represented in various forms. The polynomial representation is the most commonly used. Denote $\mathcal{R} := \mathbb{F}_q[x]/(x^m - 1)$. There exists a bijection $\sigma$ between the codewords in $\mathcal{C}$ and the polynomials in $\mathcal{R}$, establishing a one-to-one correspondence:
\begin{align*}
    \sigma: \mathcal{C} &\longrightarrow \mathcal{R}, \\
    \mathbf{c} = (c_0, c_1, \ldots, c_{m-1}) &\longmapsto c(x) = c_0 + c_1x + \cdots + c_{m-1}x^{m-1}.
\end{align*}

In fact, an $[m, k]$ cyclic code $\mathcal{C}$ is an ideal in the ring $\mathcal{R}$. This ideal is generated by a unique divisor $g(x)$ of $x^m - 1$, and is denoted as $\mathcal{C} = \langle g(x) \rangle$. Let $h(x) = \frac{x^m - 1}{g(x)}$. Then $g(x)$ and $h(x)$ are called the generator polynomial and parity-check polynomial of $\mathcal{C}$, respectively.

For a codeword $\mathbf{c}=(c_0,c_1,\dots,c_{m-1})\in \mathcal{C}$, the Hamming weight of $\mathbf{c}$ is $$w(\mathbf{c})=\mid\{i\mid c_i\neq 0, 0\leq i\leq m-1\}\mid.$$ The minimum Hamming distance of $\mathcal{C}$ is $d=\min\{w(\mathbf{c})\mid \mathbf{c}\in \mathcal{C},\mathbf{c}\neq \mathbf{0}\}$. When $m$ is even, the symplectic weight of $\mathbf{c}$ is $$w_{s}(\mathbf{c})=\mid\{i\mid (c_i,c_{\frac{m}{2}+i})\neq (0,0), 0\leq i\leq \frac{m}{2}-1\}\mid.$$ The minimum symplectic distance of $\mathcal{C}$ is $d_s=\min\{w_s(\mathbf{c})\mid \mathbf{c}\in C,\mathbf{c}\neq \mathbf{0}\}$.

The Euclidean inner product between two vectors $\mathbf{u} = (u_1, u_2, \ldots, u_m)$ and $\mathbf{v} = (v_1, v_2, \ldots, v_m)$ in $\mathbb{F}_{q}^m$ is given by
\[
\langle \mathbf{u}, \mathbf{v} \rangle_E = \sum_{i=1}^{m} u_i v_i,\] 
and the Euclidean dual code of $\mathcal{C}\subseteq \mathbb{F}_{q}^m$ is $\mathcal{C}^{\perp_E}=\{\mathbf{u}\in\mathbb{F}_{q}^m \mid \langle \mathbf{u},\mathbf{c}\rangle_E=0, \forall \mathbf{c}\in \mathcal{C}\}.$

The symplectic inner product for vectors $\mathbf{u}, \mathbf{v} \in \mathbb{F}_q^{2m}$ is defined as
\[
    \langle \mathbf{u}, \mathbf{v} \rangle_S = \sum_{i=1}^{m} u_i v_{m+i} - u_{m+i} v_i,
\]
and the symplectic dual code of $\mathcal{C}\subseteq \mathbb{F}_{q}^{2m}$ is $\mathcal{C}^{\perp_S}=\{\mathbf{u}\in\mathbb{F}_{q}^{2m} \mid \langle \mathbf{u},\mathbf{c}\rangle_S=0, \forall \mathbf{c}\in \mathcal{C}\}$.

The Hermitian inner product for vectors $\mathbf{u}, \mathbf{v} \in \mathbb{F}_{q^2}^m$ is given by
\[
    \langle \mathbf{u}, \mathbf{v} \rangle_H = \sum_{i=1}^{m} u_i^q v_i,
\]
and the Hermitian dual code of $\mathcal{C}\subseteq \mathbb{F}_{q^2}^m$ is $\mathcal{C}^{\perp_H}=\{\mathbf{u}\in\mathbb{F}_{q^2}^{m} \mid \langle \mathbf{u},\mathbf{c}\rangle_H=0, \forall \mathbf{c}\in \mathcal{C}\}$.

If $\mathcal{C}\subseteq \mathcal{C}^{\perp_E}$ (resp. $\mathcal{C}\subseteq \mathcal{C}^{\perp_S}$, $\mathcal{C}\subseteq \mathcal{C}^{\perp_H}$),  $\mathcal{C}$ is called a Euclidean (resp. symplectic, Hermitian) self-orthogonal code. If $\mathcal{C}^{\perp_E}\subseteq \mathcal{C}$ (resp. $\mathcal{C}^{\perp_S}\subseteq \mathcal{C}$, $\mathcal{C}^{\perp_H}\subseteq \mathcal{C}$),  $\mathcal{C}$ is called a Euclidean (resp. symplectic, Hermitian) dual-containing code. It is easy to see that $\mathcal{C}$ is dual-containing if and only if its dual code is self-orthogonal under the corresponding type of inner products.

 To facilitate the discussion of self-orthogonality and dual-containing properties of QC codes, we begin by introducing some notations.

 For $k(x) = k_0 + k_1x + \cdots + k_{m-1}x^{m-1}\in\mathcal{R}$,  denote
\begin{align*}
    k^{[q]}(x) &= k_0^q + k_1^q x + \cdots + k_{m-1}^q x^{m-1};\quad\quad 
    k^*(x) = x^{\deg(k(x))} k(x^{-1})\pmod{x^m - 1}; \\
    \overline{k}(x) &= k_0 + k_{m-1}x + k_{m-2}x^2 + \cdots + k_1 x^{m-1} = k(x^{-1}) \pmod{x^m - 1}; \\
    k^\perp(x) &= f(0)^{-1}{f^*(x)}, f(x)=\frac{x^m - 1}{\gcd(k(x), x^m - 1)};\quad\quad 
    k^{\perp_H}(x) = (k^{[q]}(x))^{\perp}.
\end{align*}
\begin{lemma}\label{lemmakx}
Let the symbols be the same as the above. Then we have the following several simple properties:
\begin{itemize}[leftmargin=2em]
    \item[(1)] $k^\perp(x)$ is the generator polynomial of the Euclidean dual code of the cyclic code $\langle k(x) \rangle$ and $k^{\perp_H}(x)$ is the generator polynomial of the Hermitian dual code of $\langle k(x) \rangle$. In addition, for any $a(x), b(x) \in \mathcal{R}$, the following two orthogonality relations hold:
\begin{align*}
    \langle a(x)k(x), b(x)k^\perp(x) \rangle_E = 0,\quad
    \langle a(x)k(x), b(x)k^{\perp_H}(x) \rangle_H = 0;
\end{align*}
    \item[(2)] $\overline{k}(x) = x^{m - \deg(k(x))} k^*(x)$;
    \item[(3)] $\overline{\overline{k_1}}(x) = k_1(x)$,  $\overline{k_1(x) k_2(x)} = \overline{k_1}(x) \overline{k_2}(x)$, and $\overline{k_1(x)+k_2(x)} = \overline{k_1}(x) +\overline{k_2}(x)$ for any $k_1(x), k_2(x) \in \mathcal{R}$;
    \item[(4)] $(k_1(x) k_2(x))^{[q]} = k_1^{[q]}(x) k_2^{[q]}(x)$ and $(k_1(x)+k_2(x))^{[q]} = k_1^{[q]}(x)+k_2^{[q]}(x)$ for any $k_1(x), k_2(x) \in \mathcal{R}$;
    \item[(5)] $\overline{f}^{[q]} (x)= \overline{f^{[q]}(x)}$;
   \item[(6)] $k^{\perp_H}(x) = (k^{[q]}(x))^{\perp}=(k^{\perp}(x))^{[q]}$.
\end{itemize}
\end{lemma}
\begin{proof}
    The verification of the first five statements is straightforward and thus omitted. We provide a detailed proof only for the final one.
    Our primary objective is to prove the following equation
\[
(k^{[q]}(x))^{\perp} = \big(k^\perp(x)\big)^{[q]}.
\]
Recall that
\[
k^\perp(x) = f(0)^{-1} f^*(x), \quad \text{where} \quad f(x) = \frac{x^m - 1}{\gcd(k(x), x^m - 1)}.
\]
Then,
\[
\big(k^\perp(x)\big)^{[q]} = f(0)^{-q} (f^*(x))^{[q]}.
\]
Note that
\[
  \frac{x^m - 1}{\gcd(k^{[q]}(x), x^m - 1)}=\frac{x^m - 1}{\gcd(k(x), x^m - 1)^{[q]}}=\left(\frac{x^m - 1}{\gcd(k(x), x^m - 1)}\right)^{[q]}=f^{[q]}(x)\] and  \((f^*(x))^{[q]} = \big(f^{[q]}(x)\big)^*.\)
Also, \(f^{[q]}(0) = f(0)^q\). Thus
\[
f(0)^{-q} (f^{[q]}(x))^* = f^{[q]}(0)^{-1} (f^{[q]}(x))^* = (k^{[q]}(x))^{\perp}.
\]
Therefore,
\[
(k^{[q]}(x))^{\perp} = \big(k^\perp(x)\big)^{[q]},
\]
which proves the result.
\end{proof}
As with cyclic codes, QC codes can also be expressed in polynomial form. Let $\mathcal{C}$ be a QC code of length $ml$ and index $l$ over $\mathbb{F}_q$. For any codeword $\mathbf{c}=(c_{0,0}, \ldots, c_{l-1,0}, c_{0,1}, \ldots, c_{l-1,1}, \ldots, c_{0,m-1}, \ldots, c_{l-1,m-1}) \in \mathcal{C}$, a map is defined as follows:
\[
\phi: \mathbb{F}_q^{ml} \longrightarrow \mathcal{R}^l,
\]
\[
\phi(\mathbf{c}) = (c_0(x), c_1(x), \dots, c_{l-1}(x)) \in \mathcal{R}^l,
\]
where \( c_i(x) = \sum\limits_{j=0}^{m-1} c_{i,j} x^j \in \mathcal{R} \) for each $i = 0, 1, \dots, l-1$. 
Note that $\mathcal{C}$ is a linear subspace of $\mathbb{F}_{q}^{ml}$ and it is invariant under cyclic shifts of codewords by $l$ coordinates. So, $\mathcal{C}$ corresponds exactly to an $\mathcal{R}$-submodule of $\mathcal{R}^l$.
\begin{lemma}\cite{Ling2001}\label{transform}
    For any $\mathbf{u}, \mathbf{v} \in \mathbb{F}_q^{ml}$, the Euclidean inner product of $\mathbf{v}$  with $\mathbf{u}$ and its all $l$ right cyclic shifts is zero if and only if $\phi(\mathbf{u}) * \phi(\mathbf{v}) = \sum\limits_{i=0}^{l-1} u_i(x) \overline{v_i}(x)=0$.
\end{lemma}
For the polynomial $k(x) = k_0 + k_1x + k_2x^2 + \dots + k_{m-1}x^{m-1} \in \mathcal{R}$, denote \([k(x)] = (k_0, k_1, \dots, k_{m-1})\). Then \( k(x) \) uniquely determines the following circulant matrix
\[
K =
\begin{bmatrix}
k_0 & k_1 & k_2 & \cdots & k_{m-1} \\
k_{m-1} & k_0 & k_1 & \cdots & k_{m-2} \\
\vdots & \vdots & \vdots & \ddots & \vdots \\
k_1 & k_2 & k_3 & \cdots & k_0
\end{bmatrix}=
\begin{bmatrix}
[k(x)]\\
[xk(x)]\\
\vdots\\
[x^{m-1}k(x)]
\end{bmatrix}.
\]
Let the matrix
\[
G =
\begin{bmatrix}
K_{11} & K_{12} & \cdots & K_{1l} \\
K_{21} & K_{22} & \cdots & K_{2l} \\
\vdots & \vdots & \ddots & \vdots \\
K_{t1} & K_{t2} & \cdots & K_{tl}
\end{bmatrix},
\]
where each \( K_{ij} \) is a circulant matrix corresponding to a polynomial \( k_{ij}(x) \) in $\mathcal{R}$, \( 1 \leq i \leq t\) and \(1 \leq j \leq l \). 
A QC code whose generator matrix has the form \( G \) is called a \( t \)-generator QC code with index \( l \) and its generating set can be expressed as \( \{ \vec{K}_1, \vec{K}_2, \dots, \vec{K}_{t} \} \), where $\vec{K}_i = ([k_{i1}(x)], [k_{i2}(x)], \dots, [k_{il}(x)])$,  $1 \leq i \leq t$.

Let $\mathcal{C}_1$ and $\mathcal{C}_2$ be two $t$-generator QC codes of length $ml$ and index $l$ over $\mathbb{F}_q$ with the generating sets $\{ \vec{a}_1, \vec{a}_2, \dots, \vec{a}_{t} \}$ and $\{ \vec{b}_1, \vec{b}_2, \dots, \vec{b}_{t} \}$, respectively. If $ \vec{b}_1, \vec{b}_2, \dots, \vec{b}_{t}$ can be linearly represented by $ \vec{a}_1, \vec{a}_2, \dots, \vec{a}_{t}$, we have  $\mathcal{C}_2 \subseteq \mathcal{C}_1$.

\begin{definition}\label{1-generator}\cite{Seguin2004}
   Let $\mathcal{C}$ be a 1-generator QC code of length $ml$ and index $l$ over $\mathbb{F}_q$, of the form
\[
\mathcal{C} = \left\langle \left( [a_1(x)], [a_2(x)], \dots, [a_l(x)] \right) \right\rangle 
= \left\{ r(x) \left( [a_1(x)], [a_2(x)], \dots, [a_l(x)] \right) \mid r(x) \in \mathcal{R} \right\},
\]
where $\mathcal{R} = \mathbb{F}_q[x]/( x^m - 1 )$.  
Then $g(x) = \gcd\left( a_1(x), a_2(x), \dots, a_l(x), x^m - 1 \right)$, $h(x) = \frac{x^m - 1}{g(x)}$ are called the generator polynomial and  parity-check polynomial of $\mathcal{C}$ respectively. Further, the dimension of $\mathcal{C}$ is given by $\dim(\mathcal{C}) = m - \deg(g(x)) = \deg(h(x))$.
\end{definition}
An $[[n, k, d]]_q$ quantum stabilizer code is a \(q\)-ary quantum code that encodes \(k\) logical qudits into \(n\) physical qudits and can detect all errors of weight less than \(d\), where \(q\) is a power of a prime number. Such a code is defined as the joint eigenspace of an abelian subgroup of the error group acting on the Hilbert space \(\mathbb{C}^{q^n}\).

An \((a_l, a_r)\)-$[[n, k]]_q$ quantum synchronizable code is a \(q\)-ary quantum code that encodes \(k\) qudits into $n$ physical qudits and is capable of correcting misalignment of up to \(a_l\) qudits to the left and \(a_r\) qudits to the right.

\section{Necessary and sufficient conditions of $t$-generator QC codes to be self-orthogonal}\label{sec3}
In this section, we present necessary and sufficient conditions for $t$-generator QC codes to be self-orthogonal with respect to the Euclidean, symplectic, and Hermitian inner products. These conditions offer key insights into the algebraic structure of these codes and provide a theoretical foundation for the subsequent constructions of self-orthogonal QC codes.

For the sake of brevity, we present the following definition of \( \mathcal{C} \). Unless otherwise specified, the QC code discussed in this section will refer to it.
\begin{definition}\label{def1}
    $\mathcal{C}$ is a $t$-generator QC code of index $l$ over $\mathbb{F}_q$ generated by $$\{([k_{11}(x)g_1(x)],\dots,[k_{1l}(x)g_1(x)]), \dots, ([k_{t1}(x)g_t(x)],\dots,[k_{tl}(x)g_t(x)])\},$$ where $g_i(x)\mid x^m-1$, $\gcd(k_{i1}(x),k_{i2}(x),\dots,k_{il}(x),h_i(x))=1$, $h_i(x)=\frac{x^m-1}{g_i(x)}$,  $1\leq i\leq t$, and all the polynomials belong to $\mathcal{R}=\mathbb{F}_q/(x^m-1)$. 
\end{definition}
\begin{lemma}\cite{Galindo2018}\label{Euclidean}
    Let $f(x)$, $g(x)$, and $k(x)$ be polynomials over $\mathcal{R}$. Then, the following property holds for the Euclidean inner product:
    \begin{eqnarray*}
        \langle [f(x)g(x)], [k(x)] \rangle_E = \langle [g(x)], [\overline{f}(x) k(x)] \rangle_E.
    \end{eqnarray*}
\end{lemma}
\begin{theorem}\label{E-mainth}
Let $\mathcal{C}$ be the code in Definition~\ref{def1}.
Then the necessary and sufficient conditions for $\mathcal{C}$ to be self-orthogonal 
with respect to the following inner products are given below:

\begin{enumerate}[label=\textbf{Case \Alph*.},leftmargin=*]
  \item  
        $\mathcal{C}$ is \textbf{Euclidean} self-orthogonal if and only if
        \[
    h_{r}(x)\mid \overline{g_s}(x)\sum_{j=1}^{l}k_{rj}(x)\overline{k_{sj}}(x),\quad  \text{where}\quad 1\leq r\leq s\leq t;\]

  \item  
        $\mathcal{C}$ is \textbf{Hermitian} self-orthogonal if and only if \[
    h_{r}^{[q]}(x)\mid \overline{g_s}(x)\sum_{j=1}^{l}k_{rj}^{[q]}(x)\overline{k_{sj}}(x),\quad  \text{where}\quad 1\leq r\leq s\leq t;\]

  \item 
        $\mathcal{C}$ is \textbf{Symplectic} self-orthogonal if and only if \[
    h_{r}(x)\mid \overline{g_s}(x)\sum_{j=1}^{\omega}(k_{rj}(x)\overline{k_{s(\omega+j)}}(x)-k_{r(\omega+j)}(x)\overline{k_{sj}}(x)),\]
     where $1\leq r\leq s\leq t$, $l$ is even and $\omega=\frac{l}{2}$.
\end{enumerate}

\end{theorem}

\begin{proof}
    Denote  
    \begin{align*}
    G_i=
        \begin{pmatrix}
            [k_{i1}(x)g_i(x)]&\cdots&[k_{il}(x)g_i(x)]\\
            [xk_{i1}(x)g_i(x)]&\cdots&[xk_{il}(x)g_i(x)]\\
            \vdots&\vdots&\vdots\\
            [x^{m-\deg(g_i)-1}k_{i1}(x)g_i(x)]&\cdots&[x^{m-\deg(g_i)-1}k_{il}(x)g_i(x)]
        \end{pmatrix},
    \end{align*}
    $1\leq i\leq t$. Then the generator matrix of $\mathcal{C}$ is $G = 
    \begin{bmatrix} G_1 \\ G_2 \\ \vdots \\ G_t \end{bmatrix}$. Note that $\mathcal{C}$ is Euclidean self-orthogonal if and only if $GG^T=0$, i.e., 
    \begin{align}
        \langle([x^{i_r}&k_{r1}(x)g_r(x)],[x^{i_r}k_{r2}(x)g_r(x)],\cdots,[x^{i_r}k_{rl}(x)g_r(x)]),\nonumber\\
        &([x^{j_s}k_{s1}(x)g_s(x)],[x^{j_s}k_{s2}(x)g_s(x)],\cdots,[x^{j_s}k_{sl}(x)g_s(x)])\rangle_{E}=0,\label{eq3.1}
    \end{align}
    for all $0\leq i_r\leq m-\deg(g_r(x))-1$, $0\leq j_s\leq m-\deg(g_s(x))-1$, $1\leq r,s\leq t$. By Lemma \ref{transform}, Eq.~\eqref{eq3.1} holds if and only if 
    \[x^{i_{r}+m-j_{s}}g_r(x)\overline{g_s}(x)\sum_{u=1}^{l}k_{ru}(x)\overline{k_{su}}(x)\equiv0\mod {x^m-1},\] i.e., $h_{r}(x)\mid \overline{g_s}(x)\sum_{u=1}^{l}k_{ru}(x)\overline{k_{su}}(x)$, 
    for any $r,s\in\{1,2,\cdots,t\}$. Due to the symmetry of the Euclidean inner product, it is sufficient to restrict our attention to the case of $1\leq r\leq s\leq t$. 
    
    Under the Hermitian inner product, similar necessary and sufficient conditions for the self-orthogonality of QC codes can also be established. The only modification required is to replace \( GG^T = 0 \) with \( G^{[q]}G^T = 0 \), where \( G^{[q]} \) denotes the matrix obtained by raising each entry of \( G \) to the \( q \)-th power. We omit the proof of this part here.
    
    For notational convenience, polynomial expressions like $g_r(x)$ are abbreviated as $g_r$, as long as the meaning remains unambiguous.
 Then \( \mathcal{C} \) is symplectic self-orthogonal if and only if the symplectic inner product of any two rows of the generator matrix is zero, that is,
    \begin{align*}
       \langle([x^{i_r}k_{r1}g_r],[x^{i_r}k_{r2}g_r],\cdots,[x^{i_r}k_{rl}g_r]),
         ([x^{j_s}k_{s1}g_s],[x^{j_s}k_{s2}g_s],\cdots,[x^{j_s}k_{sl}g_s])\rangle_{S}=0,
    \end{align*}
    for all $0\leq i_r\leq m-\deg(g_r(x))-1$, $0\leq j_s\leq m-\deg(g_s(x))-1$, and $1\leq r,s\leq t$. This is equivalent to
    \begin{align}
        \langle([x^{i_r}&k_{r1}g_r],\cdots,[x^{i_r}k_{r\omega}g_r]),([x^{j_s}k_{s(\omega+1)}g_s],\cdots,[x^{j_s}k_{sl}g_s])\rangle_{E}\nonumber\\&-\langle([x^{i_r}k_{r(\omega+1)}g_r],\cdots,[x^{i_r}k_{rl}g_r]),([x^{j_s}k_{s1}g_s],\cdots,[x^{j_s}k_{s\omega}g_s])\rangle_{E}=0.\label{eq7}
    \end{align}
    According to Lemma \ref{Euclidean}, Eq.~\eqref{eq7} can be transformed into
    \begin{align*}
        \langle([x^{i_r}g_r],\cdots,
        &[x^{i_r}g_r]),([x^{j_s}\overline{k_{r1}}k_{s(\omega+1)}g_s],\cdots,[x^{j_s}\overline{k_{r\omega}}k_{sl}g_s])\rangle_{E}\\-&\langle([x^{i_r}g_r],\cdots,[x^{i_r}g_r]),(
        [x^{j_s}\overline{k_{r(\omega+1)}}k_{s1}g_s],\cdots,[x^{j_s}\overline{k_{rl}}k_{s\omega}g_s])\rangle_{E}=0,
    \end{align*}
    then
    \begin{align}
        \langle([&x^{i_r}g_r],\cdots,[x^{i_r}g_r]),\nonumber\\&([x^{j_s}g_s(\overline{k_{r1}}k_{s(\omega+1)}-\overline{k_{r(\omega+1)}}k_{s1})],\cdots,[x^{j_s}g_s(\overline{k_{r\omega}}k_{sl}-\overline{k_{rl}}k_{s\omega})])\rangle_{E}=0.\label{eq8}
    \end{align}
    By Lemma \ref{transform}, Eq.~\eqref{eq8} holds for all 
    $1\leq r,s\leq t$ if and only if \[x^{i_r+m-j_s}g_r\overline{g_s}\sum_{u=1}^{\omega}(k_{ru}\overline{k_{s(\omega+u)}}-k_{r(\omega+u)}\overline{k_{su}})\equiv0\mod {x^m-1},\] i.e. $h_{r}\mid \overline{g_s}\sum_{u=1}^{\omega}(k_{ru}\overline{k_{s(\omega+u)}}-k_{r(\omega+u)}\overline{k_{su}})$ 
    for any $r,s\in\{1,2,\cdots,t\}$. Due to the skew-symmetry of the symplectic inner product, it is sufficient to restrict our attention to the case of $1\leq r\leq s\leq t$. This completes the proof.
\end{proof}
\begin{remark}
It should be noted that the above result can be viewed as a generalization of the conclusion in~\cite{BS2024}. Specifically, Case A in Theorem~\ref{E-mainth} reduces to Theorem~11 in~\cite{BS2024} when $t = 1$.
    Moreover, Ezerman et al., inspired by the idea of the Chinese Remainder Theorem, established necessary and sufficient conditions for quasi-twisted codes (a generalization of quasi-cyclic codes) to be self-orthogonal. However, those conditions are subject to a restriction on the code length $ml$, namely that $\gcd(m,q)=1$~\cite{Grassl2025}. 
\end{remark}
\begin{example}
    Let \( q = 2 \) and \( m = 21 \). Choose \( g(x) = x^{10} + x^8 + x^6 + x^4 + x^3 + 1 \), \( v_1(x) = x^3 + x + 1\), and \( v_2(x) = x^3 + x^2 \). Then we can construct a 1-generator QC code \(\mathcal{C} = \langle ([v_1(x)g(x)],\ [v_2(x)g(x))] \rangle \). Using Magma for computation, we find that \(\mathcal{C}\) is a binary linear code with parameters \([42, 11, 16]_2\), which is a best-known linear code (BKLC); see~\cite{codetables}. According to Theorem \ref{E-mainth}, the code \(\mathcal{C}\) is Euclidean self-orthogonal. 
\end{example}

\section{A new class of 2-generator QC codes}\label{sec4}

In this section, we introduce a new class of 2-generator quasi-cyclic codes with generators \( ([g_1(x)], [v_1(x)g_1(x)]) \) and \( ([v_2(x)g_2(x)], [g_2(x)]) \), and determine their dimensions.  
By applying Theorem~\ref{E-mainth} from Section~\ref{sec3}, we derive explicit necessary and sufficient conditions for these codes to be self-orthogonal with respect to the Euclidean, symplectic, and Hermitian inner products.  
We further investigate the algebraic structure of their dual codes under these three inner products and examine their dual-containing properties.  
Finally, we provide a comparative analysis between our results and those reported in related literature.


\subsection{Definition and analysis of the new QC codes }
\begin{definition}\label{definition1}
    Let $g_1(x)$, $g_2(x)$, $v_1(x)$, and $v_2(x)$ be monic polynomials in $\mathcal{R} = \mathbb{F}_q[x]/(x^m - 1)$, where $g_1(x)$ and $g_2(x)$ are divisors of $x^m - 1$, and the polynomials $v_1(x)$ and $v_2(x)$ satisfy the condition $\gcd(v_1(x)v_2(x) - 1, x^m - 1) = 1$. The 2-generator QC code $\mathcal{C}$ of length $2m$ over $\mathbb{F}_q$ is defined as the code generated by the polynomials $([g_1(x)], [v_1(x)g_1(x)])$ and $([v_2(x)g_2(x)],[g_2(x)])$.
\end{definition}
\begin{remark}
    It is worth noting that the construction in Definition~\ref{definition1} is more general than those in~\cite{Galindo2018},~\cite{Guan2022-1}, and~\cite{duchao2023}. Specifically, setting $v_2(x) = 0$ corresponds to the construction in~\cite{Galindo2018}, $v_2(x) = 1$ corresponds to that in~\cite{duchao2023}, and $v_2(x) = v_1(x)$ corresponds to the construction in~\cite{Guan2022-1}.
    In addition, the code $\mathcal{C}$ in Definition~\ref{definition1} is defined over $\mathbb{F}_{q^2}$ when the Hermitian inner product is considered.    
\end{remark}
\begin{example}\label{example2}
    Let $q=3$, $m=8$, $g_1(x)=x^3 + x^2 + x + 1$, $g_2(x)=x^6 + 2x^4 + x^2 + 2$, $v_1(x)=x^6 + 2x^4 + 2x^2 + 1$, $v_2(x)=x$, we have $\gcd(v_{1}(x)v_{2}(x)-1,x^m-1)=1$. Let $C$ be the QC code generated by $([g_1(x)],[v_1(x)g_1(x)])$ and $([v_2(x)g_2(x)],[g_2(x)])$, $C_1$ be the QC code generated by $([g_1(x)],[v_1(x)g_1(x)])$ and $(0,[g_2(x)])$, 
    $C_2$ be the QC code generated by $([g_1(x)],[v_1(x)g_1(x)])$ and $([g_2(x)],[g_2(x)])$, and $C_3$ be the QC code generated by $([g_1(x)],[v_1(x)g_1(x)])$ and $([v_1(x)g_2(x)],[g_2(x)])$. Through calculations in Magma, it is found that the parameters of the linear codes \( C, C_1, C_2 \), and \( C_3 \) are \([16, 7, 6]_3\), \([16, 7, 4]_3\), \([16, 7, 5]_3\), and \([16, 7, 4]_3\), respectively.
\end{example}
\begin{remark}
    In Example ~\ref{example2}, the code $C$ is derived from Definition~\ref{definition1}, while $C_1$, $C_2$, and $C_3$ is given by the construction in~\cite{Galindo2018},~\cite{Guan2022-1}, and~\cite{duchao2023}, respectively. It is evident that the parameters of $C$ are better than those of $C_1$, $C_2$, and $C_3$. This demonstrates that the generalized construction given in Definition~\ref{definition1} is capable of producing more and better codes.
\end{remark}
\begin{proposition}\label{generator matrix}
The QC code $\mathcal{C}$ has dimension $2m-\deg(g_1(x))-\deg(g_2(x))$. Moreover, the generator matrix of $\mathcal{C}$ is given by 
\begin{eqnarray*}
    G=
    \begin{pmatrix}
        G_1&G_{v_1} \\
        G_{v_2}&G_2
    \end{pmatrix}
    =
    \begin{pmatrix}
    [g_1(x)] & [v_1(x)g_1(x)] \\
    [xg_1(x)] & [xv_1(x)g_1(x)]\\
    \vdots & \vdots\\
    [x^{m-\deg(g_1)-1}g_1(x)] & [x^{m-\deg(g_1)-1}v_1(x)g_1(x)]\\
    [v_2(x)g_2(x)] & [g_2(x)]\\
    [xv_2(x)g_2(x)] & [xg_2(x]\\
    \vdots & \vdots\\
    [x^{m-\deg(g_2)-1}v_2(x)g_2(x)] & [x^{m-\deg(g_2)-1}g_2(x)]
    \end{pmatrix},
\end{eqnarray*}
 where $G_1$ and $G_2$ are the generator matrices of the cyclic codes $\langle g_1(x) \rangle$ and $\langle g_2(x) \rangle$, respectively. Additionally, $G_{v_1}$ and $G_{v_2}$ are $(m - \deg(g_1(x))) \times m$ and $(m-\deg(g_2(x)))\times m$ circulant matrices determined by $v_1(x)g_1(x)$ and $v_2(x)g_2(x)$, respectively.

\end{proposition}
\begin{proof}
     Let $\mathcal{C}_1$ and $\mathcal{C}_2$ be the 1-generator QC codes generated by the matrices $(G_1, G_{v_1})$ and $(G_{v_2}, G_2)$, respectively. By Definition \ref{1-generator}, we have $\dim \mathcal{C}_1 = m - \deg(g_1(x))$ and $\dim \mathcal{C}_2 = m - \deg(g_2(x))$. So, it remains to show that $\mathcal{C}_1 \cap \mathcal{C}_2 = \{\textbf{0}\}$. Suppose that there exist polynomials $a(x), b(x) \in \mathcal{R}$ such that 
\begin{eqnarray*}
    a(x)([g_1(x], [v_1(x) g_1(x)]) = b(x)([v_2(x) g_2(x)], [g_2(x)]),
\end{eqnarray*}
where $\deg(a(x)) < m - \deg(g_1(x))$ and $\deg(b(x)) < m - \deg(g_2(x))$. This leads to the system of equations:
\begin{equation*}
\begin{cases}
    a(x) g_1(x) - b(x) v_2(x) g_2(x) = 0, \\
    a(x) v_1(x) g_1(x) - b(x) g_2(x) = 0.
\end{cases}
\end{equation*}
Then 
\begin{equation*}
\begin{cases}
    x^m - 1 \mid a(x) g_1(x) - b(x) v_2(x) g_2(x), \\
    x^m - 1 \mid a(x) v_1(x) g_1(x) - b(x) g_2(x),
\end{cases}
\end{equation*}
which implies 
\begin{eqnarray*}
    x^m - 1 \mid b(x) g_2(x) (v_1(x) v_2(x) - 1).
\end{eqnarray*}
Since $\gcd(v_1(x) v_2(x) - 1, x^m - 1) = 1$, we know  that $x^m - 1 \mid b(x) g_2(x)$. Furthermore, we must have $b(x) = 0$, due to the fact that $\deg(b(x) g_2(x)) < m$. Thus we have $a(x) = 0$. Hence, we conclude that $\mathcal{C}_1 \cap \mathcal{C}_2 = \{\textbf{0}\}$. The proof is complete.
\end{proof}

\begin{theorem}\label{Eself-orthogonal}
Let $h_i(x)=\frac{x^m-1}{g_i(x)}$, $i=1,2$. 
Then $\mathcal{C}$ in Definition~\ref{definition1} is self-orthogonal 
with respect to each of the following inner products 
if and only if all the corresponding conditions in that case hold simultaneously:

\begin{enumerate}[label=\textbf{Case \Alph*.}, leftmargin=*]
  \item \textbf{Euclidean inner product.}
  \begin{align*}
    &(1.1)\quad h_1(x) \mid \overline{g_{1}}(x)\bigl(1+v_1(x)\overline{v_1}(x)\bigr),\\
    &(1.2)\quad h_1(x) \mid \overline{g_{2}}(x)\bigl(\overline{v_2}(x)+v_1(x)\bigr),\\
    &(1.3)\quad h_2(x) \mid \overline{g_{2}}(x)\bigl(1+v_2(x)\overline{v_2}(x)\bigr).
  \end{align*}

  \item \textbf{Hermitian inner product.}
  \begin{align*}
    &(2.1)\quad h_1^{[q]}(x)\mid \overline{g_1}(x)\bigl(1+v_1^{[q]}(x)\overline{v}_1(x)\bigr),\\
    &(2.2)\quad h_1^{[q]}(x)\mid \overline{g_2}(x)\bigl(\overline{v_2}(x)+v_1^{[q]}(x)\bigr),\\
    &(2.3)\quad h_2^{[q]}(x)\mid \overline{g_2}(x)\bigl(v_2^{[q]}(x)\overline{v}_2(x)+1\bigr).
  \end{align*}

  \item \textbf{Symplectic inner product.}
  \begin{align*}
    &(3.1)\quad h_1(x) \mid \overline{g_{1}}(x)\bigl(\overline{v_1}(x)-v_1(x)\bigr),\\
    &(3.2)\quad h_1(x) \mid \overline{g_{2}}(x)\bigl(1-\overline{v_2}(x)v_1(x)\bigr),\\
    &(3.3)\quad h_2(x) \mid \overline{g_{2}}(x)\bigl(v_2(x)-\overline{v_2}(x)\bigr).
  \end{align*}
\end{enumerate}

\end{theorem}

\begin{proof}
These results follow directly from Theorem~\ref{E-mainth}.
\end{proof}

\begin{lemma}\cite{Lv2020-2}\label{Hermitian}
    Let $f(x)$, $g(x)$, and $k(x)$ be polynomials over $\mathcal{R}$ (Here the base field is $\mathbb{F}_{q^2}$). Then, the following property holds for the Hermitian inner product:
    \begin{eqnarray*}
        \langle [f(x)g(x)], [k(x)] \rangle_H = \langle [g(x)], [\overline{f}^{[q]}(x) k(x)] \rangle_H.
    \end{eqnarray*}
\end{lemma}
\begin{proposition}\label{Edual}
Let $\mathcal{C}$ be the QC code in Definition~\ref{definition1}. 
Then the generators of its dual codes under three types of inner products are given as follows:

\begin{enumerate}[label=\textbf{Case \Alph*.}, leftmargin=*]
  \item 
  
  The \textbf{Euclidean} dual code $\mathcal{C}^{\perp_E}$ is the QC code generated by
  \[
    \bigl([g_{1}^\perp(x)],\,[-\overline{v}_2(x)g_{1}^\perp(x)]\bigr)
    \quad \text{and} \quad
    \bigl([-\overline{v}_1(x)g_{2}^\perp(x)],\,[g_{2}^\perp(x)]\bigr).
  \]

  \item 
  
  The \textbf{Hermitian} dual code $\mathcal{C}^{\perp_H}$ is the QC code generated by
  \[
    \bigl([g_{1}^{\perp_H}(x)],\,[-\overline{v}_2^{[q]}(x)g_{1}^{\perp_H}(x)]\bigr)
    \quad \text{and} \quad
    \bigl([-\overline{v}_1^{[q]}(x)g_{2}^{\perp_H}(x)],\,[g_{2}^{\perp_H}(x)]\bigr).
  \]

  \item
  
  The \textbf{symplectic} dual code $\mathcal{C}^{\perp_S}$ is the QC code generated by
  \[
    \bigl([g_{2}^\perp(x)],\,[\overline{v}_1(x)g_{2}^\perp(x)]\bigr)
    \quad \text{and} \quad
    \bigl([\overline{v}_2(x)g_{1}^\perp(x)],\,[g_{1}^\perp(x)]\bigr).
  \]
\end{enumerate}

\end{proposition}

\begin{proof}
    Let $\mathcal{D}_1$ denote the QC code generated by $([g_{1}^\perp(x)], [-\overline{v_2}(x)g_{1}^\perp(x)])$ and $([-\overline{v_1}(x)g_{2}^\perp(x)], [g_{2}^\perp(x)])$. It is easy to get $\gcd(\overline{v_1}(x) \overline{v_2}(x) - 1, x^m - 1) = 1$ from $\gcd(v_1(x)v_2(x) - 1, x^m - 1) = 1$, which yields that 
    \[
    \dim(\mathcal{D}_1) = 2m - \deg(g_{1}^\perp(x)) - \deg(g_{2}^\perp(x)) = \deg(g_{1}(x)) + \deg(g_2(x)).
    \]
    Let $\mathbf{c}$ be an arbitrary codeword in $\mathcal{C}$. Without loss of generality, we denote 
    \begin{align*}
    \mathbf{c} &= a_1(x)([g_1(x)], [v_1(x)g_1(x)]) + a_2(x)([v_2(x)g_2(x)], [g_2(x)]) \\
               &= ([a_1(x)g_1(x) + a_2(x)v_2(x)g_2(x)],\ [a_1(x)v_1(x)g_1(x) + a_2(x)g_2(x)]),
    \end{align*}
    where $a_1(x), a_2(x) \in \mathcal{R}$. 
    Similarly, let $\mathbf{d}_1$ be an arbitrary codeword in $\mathcal{D}_1$ and write it as 
    \begin{align*}
    \mathbf{d}_1 &= b_1(x)([g_{1}^\perp(x)], [-\overline{v_2}(x)g_{1}^\perp(x)]) + b_2(x)([-\overline{v_1}(x)g_{2}^\perp(x)], [g_{2}^\perp(x)])\\
    &= ([b_1(x)g_{1}^\perp(x) - b_2(x)\overline{v_1}(x)g_{2}^\perp(x)], [-b_1(x)\overline{v_2}(x)g_{1}^\perp(x) + b_2(x)g_{2}^\perp(x)]),
    \end{align*}
    where $b_1(x), b_2(x) \in \mathcal{R}$. 
    Then the Euclidean inner product between $\mathbf{c}$ and $\mathbf{d}_1$ is given by
    \begin{align}
        \langle\mathbf{c},\mathbf{d}_1\rangle_{E} 
        &= \langle [a_1(x)g_1(x) + a_2(x)v_2(x)g_2(x)], [b_1(x)g_{1}^\perp(x) - b_2(x)\overline{v_1}(x)g_{2}^\perp(x) ]\rangle_{E}\nonumber \\
        & + \langle [a_1(x)v_1(x)g_1(x) + a_2(x)g_2(x)], [-b_1(x)\overline{v_2}(x)g_{1}^\perp(x) + b_2(x)g_{2}^\perp(x) ]\rangle_{E}\nonumber \\
        &= \langle [a_1(x)g_1(x)], [-b_2(x)\overline{v_1}(x)g_{2}^\perp(x)] \rangle_{E} 
        + \langle [a_2(x)v_2(x)g_2(x)], [b_1(x)g_{1}^\perp(x)] \rangle_{E}\nonumber \\
        & + \langle [a_1(x)v_1(x)g_1(x)], [b_2(x)g_{2}^\perp(x)] \rangle_{E} 
        + \langle [a_2(x)g_2(x)], [-b_1(x)\overline{v_2}(x)g_{1}^\perp(x)] \rangle_{E}.\label{eq4.1}
    \end{align}
    According to Lemma~\ref{Euclidean}, it can be verified that Eq.~\eqref{eq4.1} is equal to zero. Thus, we have $\mathcal{D}_1 \subseteq \mathcal{C}^{\perp_E}$. Furthermore, as $\dim(\mathcal{D}_1) = \dim(\mathcal{C}^{\perp_E})$, we conclude that $\mathcal{D}_1 = \mathcal{C}^{\perp_E}$.
The other two cases can be proved in a similar manner, and the details are left to the reader.
\end{proof}

\begin{theorem}\label{Edual-containing}
Let $\mathcal{C}$ be the QC code in Definition~\ref{definition1}. 
Then $\mathcal{C}$ is dual-containing with respect to the following inner products 
if and only if all the corresponding conditions listed in each case hold simultaneously:

\begin{enumerate}[label=\textbf{Case \Alph*.}, leftmargin=*]
  \item For the \textbf{Euclidean inner product}, the conditions are given by
  \begin{align*}
    &(1.1)\quad g_{1}(x)\mid g_{1}^\perp(x)\bigl(1+\overline{v_2}(x)v_2(x)\bigr),\\
    &(1.2)\quad g_{1}(x)\mid g_{2}^\perp(x)\bigl(\overline{v_1}(x)+v_2(x)\bigr),\\
    &(1.3)\quad g_{2}(x)\mid g_{2}^\perp(x)\bigl(1+\overline{v_1}(x)v_1(x)\bigr).
  \end{align*}

  \item For the \textbf{Hermitian inner product}, the conditions are given by
  \begin{align*}
    &(2.1)\quad g_1(x)\mid g_{1}^{\perp_H}(x)\bigl(1+v_2(x)\overline{v}_2^{[q]}(x)\bigr),\\
    &(2.2)\quad g_1(x)\mid g_{2}^{\perp_H}(x)\bigl(\overline{v}_{1}^{[q]}(x)+v_2(x)\bigr),\\
    &(2.3)\quad g_2(x)\mid g_{2}^{\perp_H}(x)\bigl(1+v_1(x)\overline{v_1}^{[q]}(x)\bigr).
  \end{align*}

  \item For the \textbf{symplectic inner product}, the condition are given by
  \begin{align*}
    &(3.1)\quad g_{2}(x)\mid g_{2}^\perp(x)\bigl(\overline{v_1}(x)-v_1(x)\bigr),\\
    &(3.2)\quad g_{2}(x)\mid g_{1}^\perp(x)\bigl(1-\overline{v_1}(x)v_2(x)\bigr),\\
    &(3.3)\quad g_{1}(x)\mid g_{1}^\perp(x)\bigl(v_2(x)-\overline{v_2}(x)\bigr).
  \end{align*}
\end{enumerate}

\end{theorem}

\begin{proof}
    Note that $\frac{x^m - 1}{g_1^{\perp}(x)} = g_{1}^*(x)$. By replacing $\mathcal{C}$ in Theorem~\ref{E-mainth} with $\mathcal{C}^{\perp_E}$ in Proposition~\ref{Edual}, we obtain that $\mathcal{C}$ is Euclidean dual-containing if and only if the following conditions hold simultaneously:
    \begin{align*}
        &(1)\enspace g_{1}^*(x)\mid \overline{g_{1}^\perp}(x)(1+v_2(x)\overline{v_2}(x)),\\
        &(2)\enspace g_{1}^*(x)\mid \overline{g_{2}^\perp}(x)(v_1(x)+\overline{v_2}(x)),\\
        &(3)\enspace g_{2}^*(x)\mid \overline{g_{2}^\perp}(x)(1+v_1(x)\overline{v_1}(x)).
    \end{align*}
   It suffices to show that each of the three conditions stated here is equivalent to the corresponding condition in the theorem. We illustrate this using the first condition as an example; the remaining cases can be proved in a similar manner. 
    By (2) of Lemma~\ref{lemmakx}, we have 
    \begin{align*}
        \overline{g_1^{\perp}}(x)(1+v_2(x)\overline{v_2}(x))&=\overline{g_1^{\perp}(x)(1+\overline{v_2}(x)v_2(x))}\\&=x^{m-\deg(g_1^{\perp}(x)(1+\overline{v_2}(x)v_2(x)))}\left({g_1^{\perp}(x)(1+\overline{v_2}(x)v_2(x))}\right)^*.
    \end{align*}
    Furthermore, \(0\) is not a root of \(x^m - 1\), so it is also not a root of \(g_1^*(x)\). Therefore, we conclude that
 $g_{1}^*(x)\mid \overline{g_{1}^\perp}(x)(1+v_2(x)\overline{v_2}(x))$ is equivalent to $g_{1}(x) \mid g_{1}^\perp(x)(1+\overline{v_2}(x)v_2(x))$.
  The result for the Euclidean inner product case follows. By similar arguments, the results for the remaining two inner products can also be obtained and are therefore omitted.
\end{proof}
\subsection{Comparison with existing results}

In this subsection, we perform a comparative analysis between the results obtained above and those established in prior work. In particular, by assigning different values to $v_2(x)$ and giving specific examples, we can not only recover the results reported in the literature but also obtain many codes that cannot be derived from existing work. 
\subsubsection{QC codes for $v_2(x)=0$ in Definition \ref{definition1}}
\begin{definition}\label{definitionv2=0}
    Let $\mathcal{D}_1$ be a QC code over $\mathbb{F}_q$ (It is defined over $\mathbb{F}_{q^2}$ when the Hermitian inner product is considered) of length $2m$, generated by $([g_1(x)], [v_1(x)g_1(x)])$ and $(0,[g_2(x)])$, where $g_1(x)$, $g_2(x)$ and $v_1(x)$ are monic polynomials in $\mathcal{R}=\mathbb{F}_q[x]/(x^m-1)$ such that both $g_1(x)$ and $g_2(x)$ divide $x^m-1$.
\end{definition}

By substituting $v_2(x) = 0$ in Theorem \ref{Edual-containing}, we get the following three corollaries.
\begin{corollary}\label{v2=0Edual-co}
    The QC code $\mathcal{D}_1$ in Definition~\ref{definitionv2=0} is Euclidean dual-containing if and only if all the following conditions hold simultaneously:
    \begin{align*}
       (1)\enspace g_1(x)\mid g_{1}^\perp(x), \quad (2)\enspace g_1(x)\mid g_{2}^\perp(x)\overline{v_1}(x), \quad (3)\enspace g_2(x)\mid g_{2}^{\perp}(x)(1+\overline{v_1}(x)v_1(x)).
    \end{align*}
\end{corollary}
\begin{remark}
   The condition $g_1(x)\mid g_2(x)\mid g_{2}^{\perp}(x)\mid g_{1}^{\perp}(x)$, as a sufficient condition for $\mathcal{D}_1$ to be Euclidean dual-containing given in Proposition 11 of~\cite{Galindo2018}, is a special case of the necessary and sufficient conditions in Corollary~\ref{v2=0Edual-co}. This can also be seen from the following example.
\end{remark}
\begin{example}
     Let $q=2$, $m=15$, $g_1(x)=x^4 + x + 1$, $g_2(x)=x^7 + x^6 + x^4 + 1$, $v_1(x)=x^{14} + x^3 + x^2$, and $C$ be the QC code generated by $(g_1(x),v_1(x)g_1(x))$ and $(0,g_2(x))$. Upon verification, the polynomials $g_1(x)$, $g_2(x)$, and $v_1(x)$ satisfy the divisibility conditions stated in Corollary~\ref{v2=0Edual-co}, which implies that the code $C$ is Euclidean dual-containing. However, they do not satisfy the condition $g_1(x)\mid g_2(x)\mid g_{2}^{\perp}(x)\mid g_{1}^{\perp}(x)$.
\end{example}

\begin{corollary}\label{v2=0Hdual-co}
    The QC code $\mathcal{D}_1$ in Definition~\ref{definitionv2=0} is Hermitian dual-containing if and only if all the following conditions hold simultaneously:
    \begin{eqnarray*}
        (1)\enspace g_1(x)\mid g_{1}^{\perp_H}(x),\enspace (2)\enspace g_1(x)\mid g_{2}^{\perp_H}(x)\overline{v}_{1}^{[q]}(x), \enspace (3)\enspace g_2(x)\mid g_{2}^{\perp_H}(x)(1+v_1(x) \overline{v}_1^{[q]}(x)).
    \end{eqnarray*}
    
\end{corollary}
\begin{remark}   
The condition $g_1(x)\mid g_2(x)\mid g_{2}^{\perp_H}(x)\mid g_{1}^{\perp_H}(x)$, as a sufficient condition for $\mathcal{D}_1$ to be Hermitian dual-containing given in Proposition 15 of~\cite{Galindo2018}, is a special case of the necessary and sufficient conditions in Corollary~\ref{v2=0Hdual-co}. This can also be seen from the following example.
\end{remark}
\begin{example}
    Let $q=2$, $\mathbb{F}_4=\mathbb{F}_2(w)$, $m=5$, $g_1(x)=x^2 + wx + 1$, $g_2(x)=x + 1$, $v_1(x)=x^4$, and $C$ be the QC code generated by $(g_1(x),v_1(x)g_1(x))$ and $(0,g_2(x))$. Upon verification, the polynomials $g_1(x)$, $g_2(x)$, and $v_1(x)$ satisfy the divisibility conditions stated in Corollary~\ref{v2=0Hdual-co}, which implies that the code $C$ is Hermitian dual-containing. However, they do not satisfy the condition $g_1(x)\mid g_2(x)\mid g_{2}^{\perp_H}(x)\mid g_{1}^{\perp_H}(x)$.
\end{example}
\begin{corollary}\label{v2=0Sself-co}
    The QC code $\mathcal{D}_1$ in Definition~\ref{definitionv2=0} is symplectic dual-containing if and only if all the following conditions hold simultaneously.
    \begin{align*}
    (1)\enspace g_{2}(x) \mid g_{2}^\perp(x)(\overline{v_1}(x)-v_1(x)),\enspace(2)\enspace g_{2}(x) \mid g_{1}^\perp(x).
    \end{align*}
\end{corollary}
\begin{remark}
   The condition $g_1(x)\mid g_2(x)\mid g_{2}^{\perp}(x)\mid g_{1}^{\perp}(x)$, 
   as a sufficient condition for $\mathcal{D}_1$ to be symplectic dual-containing given in Theorem 5 (case (ii)) of~\cite{Galindo2018}, is a special case of the necessary and sufficient conditions in Corollary~\ref{v2=0Sself-co}. This can also be seen from the following example.
\end{remark}
\begin{example}
    Let $q=3$, $m=8$, $g_1(x)=x^5 + 2x^3 + 2x^2 + x + 2$, $g_2(x)=x + 2$, $v_1(x)=x^7 + 1$, and $C$ be the QC code generated by $(g_1(x),v_1(x)g_1(x))$ and $(0,g_2(x))$. Upon verification, the polynomials $g_1(x)$, $g_2(x)$, and $v_1(x)$ satisfy the divisibility conditions stated in Corollary~\ref{v2=0Sself-co}, which implies that the code $C$ is symplectic dual-containing. However, they do not satisfy the condition $g_1(x)\mid g_2(x)\mid g_{2}^{\perp}(x)\mid g_{1}^{\perp}(x)$.
\end{example}
\subsubsection{QC codes for $v_2(x)=1$ in Definition \ref{definition1}}
\begin{definition}\label{defv2=1}
    Let $\mathcal{D}_2$ be a QC code over $\mathbb{F}_q$ of length $2m$, generated by $([g_1(x)], [v_1(x)g_1(x)])$ and $([g_2(x)], [g_2(x)])$, where $g_1(x)$, $g_2(x)$, and $v_1(x)$ are monic polynomials in $\mathcal{R} = \mathbb{F}_q[x]/(x^m-1)$. Moreover, $g_1(x)$ and $g_2(x)$ are divisors of $x^m-1$, and $v_1(x)$ satisfies $\gcd(v_1(x)-1, x^m-1) = 1$.
\end{definition}
Substituting $v_2(x) = 1$ into Cases A and C of Theorem~\ref{Edual-containing} yields the following corollaries.
\begin{corollary}\label{v2=1Edual-co}
    The QC code $\mathcal{D}_2$ in Definition~\ref{defv2=1} is Euclidean dual-containing if and only if all of the following conditions hold simultaneously: 
    \begin{eqnarray*}
        (1)\enspace g_{1}(x) \mid 2g_{1}^\perp(x), \enspace  (2)\enspace g_{1}(x) \mid g_{2}^\perp(x)(\overline{v_1}(x)+1), \enspace (3) \enspace g_{2}(x) \mid g_{2}^\perp(x)(1+\overline{v_1}(x)v_1(x)).
    \end{eqnarray*}
    Specially, for $q = 2$, the above conditions can be simplified to $g_{1}(x) \mid g_{2}^\perp(x)(\overline{v_1}(x)+1)$ and $g_{2}(x) \mid g_{2}^\perp(x)(1+\overline{v_1}(x)v_1(x))$.
\end{corollary}
\begin{remark}
    The conditions $g_2(x)\mid g_{1}^\perp(x)$, $g_2(x)\mid g_{2}^\perp(x)$, and $\overline{v_1}(x)=v_1(x)$, as sufficient conditions for $\mathcal{D}_2$ to be Euclidean dual-containing given in Lemma 10  (It is valid only for $q=2$) of \cite{duchao2023}, are a special case of the necessary and sufficient conditions in Corollary~\ref{v2=1Edual-co} (It is valid for any prime power $q$) because $g_1(x) \mid g_2^\perp(x)\iff g_2(x) \mid g_1^\perp(x)$. This can also be seen from the following example.
\end{remark}
\begin{example}
    Let $q=2$, $m=9$, $g_1(x)=x + 1$, $g_2(x)=x^2 + x + 1$, $v_1(x)=x^8+1$, and $C$ be the QC code generated by $(g_1(x),v_1(x)g_1(x))$ and $(g_2(x),g_2(x))$. Upon verification, the polynomials $g_1(x)$, $g_2(x)$, $v_1(x)$ satisfy the divisibility conditions stated in Corollary~\ref{v2=1Edual-co}, which implies that the code $C$ is Euclidean dual-containing. However, they do not satisfy the conditions $g_2(x)\mid g_{1}^\perp(x)$, $g_2(x)\mid g_{2}^\perp(x)$, and $\overline{v_1}(x)=v_1(x)$ simultaneously.
\end{example}

\begin{corollary}\label{v2=1Sdual-co}
    The QC code $\mathcal{D}_2$ in Definition~\ref{defv2=1} is symplectic dual-containing if and only if all of the following conditions hold simultaneously:
    \begin{eqnarray*}
     (1)\enspace g_{2}(x) \mid g_{2}^\perp(x)(\overline{v_1}(x)-v_1(x)), \enspace (2) \enspace g_{2}(x) \mid g_{1}^\perp(x)(1-\overline{v_1}(x)). 
    \end{eqnarray*}
\end{corollary}
\begin{remark}
The conditions $g_1(x)\mid g_{2}^\perp(x)$ and $\overline{v_1}(x)=v_1(x)$, as sufficient conditions for $\mathcal{D}_2$ to be symplectic dual-containing stated in 
Proposition 7 of \cite{Lv2020}, are a special case of the necessary and sufficient conditions in Corollary~\ref{v2=1Sdual-co} because $g_1(x) \mid g_2^\perp(x)\iff g_2(x) \mid g_1^\perp(x)$. This can also be seen from the following example.
\end{remark}
\begin{example}
    Let $q=2$, $m=7$, $g_1(x)=x^3 + x + 1$, $g_2(x)=x + 1$, $v_1(x)=x^6 + 1$, and $C$ be the QC code generated by $(g_1(x),v_1(x)g_1(x))$ and $(g_2(x),g_2(x))$. Upon verification, the polynomials $g_1(x)$, $g_2(x)$, $v_1(x)$ satisfy the divisibility conditions stated in Corollary~\ref{v2=1Sdual-co}, which implies that the code $C$ is symplectic dual-containing. However, they do not satisfy the conditions $g_1(x)\mid g_{2}^\perp(x)$ and $\overline{v_1}(x)=v_1(x)$ simultaneously.
\end{example}
\subsubsection{QC codes for $v_1(x)=v_2(x)$ in Definition \ref{definition1}}
\begin{definition}\label{definitionv1=v2}
    Let $\mathcal{D}_3$ be the QC code over $\mathbb{F}_q$ (It is defined over $\mathbb{F}_{q^2}$ when the Hermitian inner product is considered) of length $2m$, generated by $([g_1(x)], [v(x)g_1(x)])$ and $([v(x)g_2(x)], [g_2(x)])$, where $g_1(x)$, $g_2(x)$, and $v(x)$ are monic polynomials in $\mathcal{R}=\mathbb{F}_q[x]/(x^m-1)$, such that $g_1(x)$ and $g_2(x)$ divide $x^m - 1$, and $\gcd(v^2(x) - 1, x^m - 1) = 1$.
\end{definition}
\begin{remark}
 For the definition of the QC code \( \mathcal{D}_3 \), Proposition 1 of \cite{Guan2022-2} requires that \( v(x) \) (denoted as \( t(x) \) there) is a polynomial over \( \mathcal{R}_2 = \mathbb{F}_{q^2}/(x^m - 1) \) such that \( \dim \mathcal{D}_3 = 2m - \deg(g_1(x)) - \deg(g_2(x)) \), but does not explicitly specify constraints on it. In our approach, We have included a concrete formulation of the condition, i.e., \( \gcd(v^2(x) - 1, x^m - 1) = 1 \), ensuring that \( \dim \mathcal{D}_3 = 2m - \deg(g_1(x)) - \deg(g_2(x)) \).  
\end{remark}
The following corollaries follow from Cases B and C of Theorem~\ref{Edual-containing} by setting $v_1(x) = v_2(x) = v(x)$.
\begin{corollary}\label{corollaryh}
    The QC code $\mathcal{D}_3$ in Definition~\ref{definitionv1=v2} is Hermitian dual-containing if and only if all of the following conditions hold simultaneously:
    \begin{align*}
        &(1)\enspace g_1(x)\mid g_{1}^{\perp_H}(x) (1+v(x) \overline{v}^{[q]}(x)),\\
        &(2)\enspace g_1(x)\mid g_{2}^{\perp_H}(x) (\overline{v}^{[q]}(x) + v(x)),\\  
        &(3)\enspace g_2(x)\mid g_{2}^{\perp_H}(x)(1+v(x) \overline{v}^{[q]}(x)).
    \end{align*}
\end{corollary}
\begin{remark}
    The conditions $g_1(x)\mid g_{1}^{\perp_H}(x)$, $g_1(x)\mid g_{2}^{\perp_H}(x) (\overline{v}^{[q]}(x) + v(x))$, and $g_2(x)\mid g_{2}^{\perp_H}(x)$, as sufficient conditions for $\mathcal{D}_3$ to be Hermitian dual-containing stated in Proposition~2 of~\cite{Guan2022-2}, are special cases of the necessary and sufficient conditions in Corollary \ref{corollaryh}. This can also be seen from the following example.
\end{remark}
\begin{example}
    Let $q=2$, $\mathbb{F}_4=\mathbb{F}_2(w)$, $m=5$, $g_1(x)=x^2 + wx + 1$, $g_2(x)=x^2 + wx + 1$, $v(x)=x^4 + w$, and $C$ be the QC code generated by $(g_1(x),v(x)g_1(x))$ and $(v(x)g_2(x),g_2(x))$. Upon verification, the polynomials $g_1(x)$, $g_2(x)$, $v(x)$ satisfy the divisibility conditions given in Corollary~\ref{corollaryh}, which implies that the code $C$ is Hermitian dual-containing. However, they do not satisfy the conditions $g_1(x)\mid g_{1}^{\perp_H}(x)$, $g_1(x)\mid g_{2}^{\perp_H}(x) (\overline{v}^{[q]}(x) + v(x))$, and $g_2(x)\mid g_{2}^{\perp_H}(x)$ simultaneously.
\end{example}
\begin{corollary}\label{corollarys}
    The QC code $\mathcal{D}_3$ in Definition~\ref{definitionv1=v2} is symplectic dual-containing if and only if all of the following conditions hold simultaneously:
    \begin{align*}
        &(1)\enspace g_{1}(x) \mid g_{1}^\perp(x)(v(x)-\overline{v}(x)),\\ &(2)\enspace g_{2}(x) \mid g_{1}^\perp (x)(1-\overline{v}(x)v(x)),\\ &(3)\enspace g_{2}(x) \mid g_{2}^\perp(x)(\overline{v}(x)-v(x)).
    \end{align*}
\end{corollary}
\begin{remark}
    The conditions $g_1(x) \mid g_2^\perp(x)$, $\gcd(g_1(x), g_2(x)) = 1$, and $\overline{v}(x)=v(x)$, as sufficient conditions for $D_3$
  to be symplectic dual-containing stated in Proposition 4 (It is valid only for $q=2$) of ~\cite{Guan2022-1}, are  special cases of the necessary and sufficient conditions in Corollary ~\ref{corollarys} (It is valid for any prime power $q$) because $g_1(x) \mid g_2^\perp(x)\iff g_2(x) \mid g_1^\perp(x)$. This can also be seen from the following example.
  
\end{remark}
\begin{example}
    Let $q=3$, $m=8$, $g_1(x)=x^5 + 2x^3 + 2x^2 + x + 2$, $g_2(x)= x^2 + 2x + 2$, $v(x)=x^7 + 2x^3$, and $C$ be the QC code generated by $(g_1(x),v(x)g_1(x))$ and $(v(x)g_2(x),g_2(x))$. Upon verification, the polynomials $g_1(x)$, $g_2(x)$, $v(x)$ satisfy the divisibility conditions given in Corollary~\ref{corollarys}, which implies that the code $C$ is symplectic dual-containing. However, they do not satisfy the conditions $g_1(x) \mid g_2^\perp(x)$, $\gcd(g_1(x), g_2(x)) = 1$, and $\overline{v}(x)=v(x)$ simultaneously.
\end{example}
\section{Construction of quantum stabilizer codes}\label{sec5}
This section extends the results established in the previous section to construct quantum stabilizer codes from self-orthogonal QC codes.
 With the aid of the algebraic computation software Magma ~\cite{magma}, we systematically compute the parameters of these codes, some of which match the best-known quantum error-correcting codes listed in the codetable of Grassl~\cite{codetables}. These results highlight the potential of QC codes in constructing good quantum codes.
To begin with, we provide the following lemmas.
\begin{lemma}\cite{Ashikhmin2001}\label{ECSS construction}
    Given two classical linear codes \( \mathcal{C}_1 \) and \( \mathcal{C}_2 \) of length $n$ over \( \mathbb{F}_q \) with \( \mathcal{C}_2 \subset \mathcal{C}_1 \), we can construct a stabilizer quantum code with parameters  
\[
[[n, \dim \mathcal{C}_1 - \dim \mathcal{C}_2, \min\{w(\mathcal{C}_1 \setminus \mathcal{C}_2), w(\mathcal{C}_2^\perp \setminus \mathcal{C}_1^\perp)\}]]_q.
\] 
\end{lemma}

\begin{lemma}\cite{Ketkar2006}\label{HCSS construction}
   Let \( \mathcal{C} \) be a Hermitian self-orthogonal linear code \([n,k]_{q^2}\) over \(\mathbb{F}_{q^2}\), such that every nonzero vector in \( \mathcal{C}^{\perp_H}\setminus \mathcal{C}\) has weight at least \( d^{\prime} \). Then, \( \mathcal{C} \) induces a quantum stabilizer code with parameters  
\[
[[n, n-2k, d^{\prime}]]_q.
\]
\end{lemma}
\begin{lemma}\cite{Calderbank1998}\label{SCSS construction}
    Let \( \mathcal{C} \) be a symplectic self-orthogonal \([2n, k]\) linear code over $\mathbb{F}_q$. Then, a quantum stabilizer code with parameters
\[
[[n, n-k, d^{\prime}]]_q
\]
can be constructed, where \( d^{\prime} \) denotes the minimum symplectic  weight of \( \mathcal{C}^{\perp_S} \setminus \mathcal{C} \).
\end{lemma}
According to Theorem~\ref{Eself-orthogonal} and Lemma~\ref{ECSS construction}, a quasi-cyclic (QC) code that satisfies the Euclidean self-orthogonality condition can be used to construct a corresponding quantum stabilizer code, as detailed below.

\begin{theorem}\label{EtoQEC}
    Suppose that $\mathcal{C}$ is a 2-generator QC code defined in Definition~\ref{definition1} and satisfies the following conditions:
    \begin{align*}
        &(1)\enspace h_1(x) \mid \overline{g_{1}}(x)(1+v_1(x)\overline{v_1}(x)),\\
        &(2)\enspace h_1(x) \mid \overline{g_{2}}(x)(\overline{v_2}(x)+v_1(x)),\\ 
        &(3)\enspace h_2(x) \mid \overline{g_{2}}(x)(1+v_2(x)\overline{v_2}(x)),
    \end{align*}
    where $h_i(x)=\frac{x^m-1}{g_i(x)}$, $i=1,2$.
Then, \( \mathcal{C} \) is a Euclidean self-orthogonal code with parameters \([2m, 2m - \deg(g_1) - \deg(g_2)]\). Therefore, there exists a quantum stabilizer code with the following parameters:  
\[
[[2m, 2(\deg(g_1)+\deg(g_2)-m), w(\mathcal{C}^\perp \setminus \mathcal{C})]]_q.
\]
\end{theorem}
\begin{example}
    Let $q=2$, $m=9$, $g_1(x)=x^7 + x^6 + x^4 + x^3 + x + 1$, $g_2(x)=x^8 + x^7 + x^6 + x^5 + x^4 + x^3 + x^2 + x + 1$, $v_1(x)=x^8 + x + 1$, $v_2(x)=x^8 + 1$, we have $\gcd(v_{1}(x)v_{2}(x)-1,x^m-1)=1$. Let $\mathcal{C}$ be the QC code generated by $(g_1(x),v_1(x)g_1(x))$ and $(v_2(x)g_2(x),g_2(x))$. Through calculations in Magma, it is found that the parameters of linear code $\mathcal{C}^{\perp_E}$ are $[18,14,2]_2$. Applying Theorem \ref{EtoQEC}, we construct a quantum error-correcting code with parameters $[[18, 12, 2]]_2$, which coincide with those of a best-known code listed in the codetable of Grassl~\cite{codetables}.
\end{example}
Drawing upon Theorem~\ref{Eself-orthogonal} and Lemma~\ref{HCSS construction}, we can construct a quantum stabilizer code from any QC code that satisfies the Hermitian self-orthogonal condition, as presented below.
\begin{theorem}\label{HtoQEC}
    Suppose that $\mathcal{C}$ is a 2-generator QC code defined in Definition~\ref{definition1} and satisfies the following conditions: 
    \begin{align*}
        &(1)\enspace h_1^{[q]}(x)\mid \overline{g_1}(x)(1+v_1^{[q]}(x)\overline{v}_1(x)), \\
        &(2)\enspace h_1^{[q]}(x)\mid \overline{g_2}(x)(\overline{v}_2(x)+v_1^{[q]}(x)),\\
        &(3)\enspace h_2^{[q]}(x)\mid \overline{g_2}(x)(v_2^{[q]}(x)\overline{v}_2(x)+1),
    \end{align*}
    where $h_i(x)=\frac{x^m-1}{g_i(x)}$, $i=1,2$.
    Then, $\mathcal{C}$ is a Hermitian self-orthogonal code with parameters $[2m, 2m-\deg(g_1)-\deg(g_2)]_{q^2}$. Therefore, there exists a quantum stabilizer code with the following parameters: 
    \[[[2m,2(\deg (g_1)+\deg(g_2)-m),w(\mathcal{C}^{\perp_H}\setminus \mathcal{C})]]_q.\]
\end{theorem}
\begin{example}
    Let $q^2=4$, $m=3$, $\mathbb{F}_4=\mathbb{F}_2(w)$,  $g_1(x)=x^2 + wx + w + 1$, $g_2(x)=x^2 + (w + 1)x + w$, $v_1(x)=x^2$, $v_2(x)=x^2 + w + 1$, we have $\gcd(v_{1}(x)v_{2}(x)-1,x^m-1)=1$. Let $\mathcal{C}$ be the QC code generated by $(g_1(x),v_1(x)g_1(x))$ and $(v_2(x)g_2(x),g_2(x))$. Through calculations in Magma, it is found that the parameters of linear code $\mathcal{C}^{\perp_H}$ are $[6,2,2]_2$. Applying Theorem \ref{HtoQEC}, we construct a quantum error-correcting code with parameters $[[6, 2, 2]]_2$, which coincide with those of a best-known code listed in the codetable of Grassl~\cite{codetables}.
\end{example}
Based on Theorem~\ref{Eself-orthogonal} and Lemma~\ref{SCSS construction}, any QC code that satisfies the symplectic self-orthogonal condition can be used to construct a quantum stabilizer code, as stated below.
\begin{theorem}\label{StoQEC}
Suppose that $\mathcal{C}$ is a 2-generator QC code defined in Definition~\ref{definition1} and satisfies the following conditions: 
\[
 h_1(x) \mid \overline{g_1}(x)(\overline{v_1}(x)-v_1(x)), \enspace h_1(x) \mid \overline{g_2}(x)(1-\overline{v_2}(x)v_1(x)), \enspace h_2(x) \mid \overline{g_2}(x)(v_2(x)-\overline{v_2}(x)).
\]
where $h_i(x)=\frac{x^m-1}{g_i(x)}$, $i=1,2$.
Then, $\mathcal{C}$ is a symplectic self-orthogonal code with parameters $[2m, 2m-\deg(g_1)-\deg(g_2)]$. Therefore, there exists a quantum stabilizer code with the following parameters:  
\[[[m, \deg(g_1)+\deg(g_2)-m, w_s(\mathcal{C}^{\perp_S}\setminus \mathcal{C})]]_q.\]
\end{theorem}

\begin{example}
    Let $q=2$, $m=6$, $g_1(x)=x^4 + x^3 + x + 1$, $g_2(x)=x^5 + x^4 + x^3 + x^2 + x + 1$, $v_1(x)=x^5$, $v_2(x)=x^5 + x$, we have $\gcd(v_{1}(x)v_{2}(x)-1,x^m-1)=1$. Let $\mathcal{C}$ be the QC code generated by $(g_1(x),v_1(x)g_1(x))$ and $(v_2(x)g_2(x),g_2(x))$. Through calculations in Magma, it is found that the parameters of linear code $\mathcal{C}^{\perp_S}$ are $[12,9,2]_2$. Applying Theorem \ref{StoQEC}, we construct a quantum error-correcting code with parameters $[[6, 3, 2]]_2$, which coincide with those of a best-known code listed in the codetable of Grassl~\cite{codetables}.
\end{example}
Through extensive numerical experiments, we have constructed a large number of quantum stabilizer codes with excellent parameters. Some representative examples of relatively simple cases are listed in the tables below.
 For simplicity, we only write the coefficients of the polynomials instead of the full polynomials. For example, the polynomial \( x^3 + 2x + 1 \) is represented by \((1 201)\).

\begin{table}[h]
\begin{center}
\begin{minipage}{\textwidth}
\caption{QECCs Constructed from Euclidean self-orthogonal 2-Generator QC Codes}\label{Etab}%
\begin{tabular}{l l  l l}\hline
$q$ &$m$ &Polynomials $g_1$,$g_2$,$v_1$,$v_2$ & QECCs  \\

\hline
2&9&$(1 1 1 1 1 1 1 1 1),(1 1 0 1 1 0 1 1),(0 0 0 0 0 1),(0 1 0 0 0 1)$&$[[18, 12, 2]]_2$\\
\hline
2&15&$(1 1 0 1 1 0 1 1 0 1 1 0 1 1), (1 1 1 1 1 1 1 1 1 1 1 1 1 1 1), (1 1 0 0 0 1), (1 0 0 0 0 1)
$&$[[30, 24, 2]]_2$\\
\hline 
3&6&$(2 1 2 1 2 1),(1 1 1 1 1 1),(0 0 0 0 0 1),(1 0 0 0 0 1)$&$[[12, 8, 2]]_3$\\

\hline   
\end{tabular}
\end{minipage}
\end{center}
\end{table}

\begin{table}[h]\label{tableH}
\begin{center}
\begin{minipage}{\textwidth}
\caption{QECCs Constructed from Hermitian self-orthogonal 2-Generator QC Codes}
\begin{tabular}{l l  l l}\hline
$q^2$ &$m$ &Polynomials $g_1$,$g_2$,$v_1$,$v_2$ & QECCs \\
\hline
4 & 3 & $(w+1 w 1),(w 1),(0 0 1),(w+1 0 1)$

& $[[6, 0, 4]]_2$\\

\hline
4&5&$(1 w+1 1),(1 w w 1),(0 0 0 0 1),(1 0 0 0 1)
$& $[[10, 0, 4]]_2$\\

 \hline  
\end{tabular}
\end{minipage}
\end{center}
\end{table}

\begin{table}[h]
\begin{center}
\begin{minipage}{\textwidth}
\caption{QECCs Constructed from symplectic self-orthogonal 2-Generator QC Codes}\label{stab}%
\begin{tabular}{l l  l l}\hline
$q$ &$m$ &Polynomials $g_1$,$g_2$,$v_1$,$v_2$ & QECCs  \\
\hline
2 & 9 & 
$(111111111),(1001),(0 0 0 0 0 0 0 0 1),(0 1 0 0 0 0 0 0 1)$& $[[9, 2, 3]]_2$\\
 
\hline
2&9&$(111111111), (11011011), (000000001), (001000001)
$&$[[9, 6, 2]]_2$\\

\hline
 2 &13 & $(11), (1111111111111), (0100100010001), (0000000000001)
$ &$[[13, 0, 5]]_2$ \\ 
\hline
3&7&$(21), (1111111), (0100001), (0000001)
$ &$[[7, 0, 4]]_3$\\

\hline   
\end{tabular}
\end{minipage}
\end{center}
\end{table}

\section{Construction of quantum synchronizable codes}\label{sec6}
In this section, we construct quantum synchronizable codes using the results on 2-generator quasi-cyclic codes over $\mathbb{F}_q$ obtained in Section \ref{sec4}.  Here, we still assume that $p$ is the characteristic of $\mathbb{F}_q$.

\begin{lemma}\label{qc construction}\cite{duchao2023}
Let $\mathcal{C}_1$ be a 2-generator QC code over $\mathbb{F}_q$ with parameters $[2m, k_1, d_1]$, generated by $\{([a_{1,1}(x)],[a_{1,2}(x)]), ([a_{2,1}(x)],[a_{2,2}(x)])\}$. Similarly, let $\mathcal{C}_2$ be another 2-generator QC code over $\mathbb{F}_q$ with parameters $[2m, k_2, d_2]$, generated by $\{([b_{1,1}(x)], [b_{1,2}(x)]),([b_{2,1}(x)],[b_{2,2}(x)])\}$. Suppose that $\mathcal{C}_1^\perp \subset \mathcal{C}_1 \subset \mathcal{C}_2$, and that for some integer $1 \leq j \leq 2$, the greatest common divisor $\eta(x) = \gcd(a_{1,j}(x), a_{2,j}(x))$ is nontrivial, and satisfies $\eta(x) = f(x) b_{1,j}(x)$, where $b_{1,j}(x) = \gcd(b_{1,j}(x), b_{2,j}(x))$. Then, for any pair of non-negative integers $a_l$ and $a_r$ satisfying $a_l + a_r < \text{ord}(f(x))$, there exists an $(a_l, a_r)$-quantum synchronizable code with parameters $[[2m + a_l + a_r, 2(k_1 - m)]]_q$, which can correct at least up to $\lfloor\frac{d_1 - 1}{2}\rfloor$ phase errors and at least up to $\lfloor\frac{d_2 - 1}{2}\rfloor$ bit errors.
\end{lemma}
Here, $\text{ord}(f(x))$ refers to the order of $f(x)$, namely the smallest positive integer $\tau$ such that $f(x) \mid x^{\tau} - 1$ over the given field. When $\operatorname{ord}(f(x))$ in the above lemma attains its maximum possible value, the resulting quantum synchronizable code achieves the best possible tolerance against misalignment.
We now focus on the cases where $m = p^t$ and $m = lp^t$, with $l, t \geq 1$ and $\gcd(l, p) = 1$, providing explicit constructions for codes that attain this level of tolerance.

\subsection{Quantum synchronizable codes from QC codes with $m = p^t$
}
\begin{lemma}\label{pt}
    Let $\mathcal{C}$ be the QC code generated by $([(x-1)^{r_1}],[v_1(x)(x-1)^{r_1}])$ and $([v_2(x)(x-1)^{r_2}],[(x-1)^{r_2}])$, where $v_1(x),v_2(x)\in \mathbb{F}_q/(x^m-1)$, $\gcd(v_1(x)v_2(x)-1,x^{p^t}-1)=1$, $0<r_1< r_2<\frac{p^t-1}{2}$. Then $\mathcal{C}^{\perp_E}$ is the QC code generated by $([(x-1)^{p^t-r_1}],[-\overline{v}_2(x)(x-1)^{p^t-r_1}])$ and $([-\overline{v}_1(x)(x-1)^{p^t-r_2}],[(x-1)^{p^t-r_2}])$, moreover, we have $\mathcal{C}^{\perp_E}\subset \mathcal{C}$.
\end{lemma}
\begin{proof}
    Denote \( g_i(x)=(x-1)^{r_i} \), then \( g_i^\perp(x)=(x-1)^{p^t-r_i} \) for \( i=1,2 \). According to Proposition~\ref{Edual}, the Euclidean dual code \( \mathcal{C}^{\perp_E} \) of \( \mathcal{C} \) is generated by $([(x-1)^{p^t-r_1}],[-\overline{v}_2(x)(x-1)^{p^t-r_1}])$ and $([-\overline{v}_1(x)(x-1)^{p^t-r_2}],[(x-1)^{p^t-r_2}])$. Furthermore, as a result of the constraint \( 0 < r_1 \leq r_2 < \frac{p^t - 1}{2} \), the polynomials \( g_1(x), g_2(x), v_1(x), v_2(x) \) satisfy the divisibility conditions stated in Theorem~\ref{Edual-containing}, implying that \( \mathcal{C} \) is Euclidean dual-containing. This completes the proof.

\end{proof}
    \begin{theorem}\label{qsc1}
    Let \( \mathcal{C}_1 \) be a \([2p^t, k_1, d_1]\) QC code generated by \( ([(x-1)^{r_1}], [v_1(x)(x-1)^{r_1}]) \) and \( ([v_2(x)(x-1)^{r_2]}, [(x-1)^{r_2}]) \), where \( v_1(x) \) and \( v_2(x) \) are monic polynomials in \( \mathcal{R} \)
, \( \gcd(v_1(x)v_2(x)-1, x^{p^t}-1) = 1 \), and \( 0 < r_1 < r_2 < \frac{p^t-1}{2} \).  
    Similarly, let \( \mathcal{C}_2 \) be another  \([2p^t, k_2, d_2]\) QC code generated by \( ([(x-1)^{s_1}], [v_1(x)(x-1)^{s_1}]) \) and \( ([v_2(x)(x-1)^{s_2}], [(x-1)^{s_2}]) \), where \( 0 < s_1 < r_1 \) and \( 0 < s_1 < s_2 < r_2 \).  
    Suppose \( r_1 - s_1 > p^{t-1} \). Then, for any non-negative integers \( a_l, a_r \) satisfying \( a_l + a_r < p^t \), there exists a quantum synchronizable code with parameters  
    \[
    (a_l,a_r)-[[2p^t + a_l + a_r, 2p^t - 2r_1 - 2r_2]]_q,
    \]  
    which can correct at least up to $\lfloor\frac{d_1 - 1}{2}\rfloor$ phase errors and at least up to $\lfloor\frac{d_2 - 1}{2}\rfloor$ bit errors.
\end{theorem}

\begin{proof}
    Since \( 0 < r_1 \leq r_2 < \frac{p^t-1}{2} \), it follows from Lemma \ref{pt} that \( \mathcal{C}^{\perp_E} \subset \mathcal{C} \). Define \( f_1(x) = (x-1)^{r_1 - s_1} \) and \( f_2(x) = (x-1)^{r_2 - s_2} \), then we have  
    \begin{align*}
        \begin{pmatrix}
            [(x-1)^{r_1}] & [v_1(x)(x-1)^{r_1}] \\
            [v_2(x)(x-1)^{r_2}] & [(x-1)^{r_2}]
        \end{pmatrix}
        \\=
        \begin{pmatrix}
            [f_1(x)] & \mathbf{0} \\
            \mathbf{0} & [f_2(x)]
        \end{pmatrix}
        &
        \begin{pmatrix}
            [(x-1)^{s_1}] & [v_1(x)(x-1)^{s_1}] \\
            [v_2(x)(x-1)^{s_2}] & [(x-1)^{s_2}]
        \end{pmatrix}.
    \end{align*}
    This equation implies that \( \mathcal{C}_1 \subset \mathcal{C}_2 \). Furthermore, we observe that  $\eta(x) = \gcd((x-1)^{r_1}, v_2(x)(x-1)^{r_2}) = (x-1)^{r_1}$, $(x-1)^{s_1} = \gcd((x-1)^{s_1}, v_2(x)(x-1)^{s_2})$ and $(x-1)^{r_1} = (x-1)^{r_1 - s_1} (x-1)^{s_1}$. Moreover, the dimensions of the codes are given by  
    \[
    \dim \mathcal{C}_1 = k_1 = 2p^t - r_1 - r_2, \quad \dim \mathcal{C}_2 = k_2 = 2p^t - s_1 - s_2.
    \]
    Since \( r_1 - s_1 > p^{t-1} \), the order of \( (x-1)^{r_1 - s_1} \) is \( p^t \), which ensures that the best attainable tolerance against misalignment is \( p^t \).  
    Applying Lemma \ref{qc construction}, we obtain a quantum synchronizable code with parameters  
    \[
    (a_l,a_r)-[[2p^t + a_l + a_r, 2p^t - 2r_1 - 2r_2]]_q,
    \]  
    where \( a_l + a_r < p^t \).  
    This completes the proof.
\end{proof}
Building on the theoretical foundations above, we now present examples of constructing quantum synchronizable codes from QC  codes.
\begin{example}
    Let $q = p = 3$, $t = 3$, $v_1(x) =x^3 + x^2 + 1$, and $v_2(x) = x^5 + 2x^3 + x^2 + 2x + 1$. We choose $r_1 = 12$, $r_2 = 13$, $s_1 = 2$, and $s_2 = 5$. Then, $\mathcal{C}_1$ is generated by
$\bigl( [ (x-1)^{12} ], [ v_1(x)(x-1)^{12} ] \bigr) \enspace\text{and}\enspace   
        \bigl( [ v_2(x)(x-1)^{13} ], [ (x-1)^{13} ] \bigr)$.  Similarly, $\mathcal{C}_2$ is generated by
    $\bigl( [ (x-1)^{2} ], [ v_1(x)(x-1)^{2} ] \bigr)$ and $\bigl( [ v_2(x)(x-1)^{5} ], [ (x-1)^{5} ] \bigr)$. By Theorem~\ref{qsc1}, we obtain a quantum synchronizable code with parameters $(a_l, a_r)-[[54 + a_l + a_r, 4]]_3$, where $a_l + a_r < 27$. 

\end{example}

\begin{example}
    Let $q = p = 5$, $t = 2$, $v_1(x) =x^3 + x^2 + 2x + 3$, and $v_2(x) = x^5 + 2x^4 + x^3 + 2x^2 + x$. We choose $r_1 = 9$, $r_2 = 11$, $s_1 = 3$, and $s_2 = 4$. Then, $\mathcal{C}_1$ is generated by
$\bigl( [ (x-1)^{9} ], [ v_1(x)(x-1)^{9} ] \bigr)$ and   
        $\bigl( [ v_2(x)(x-1)^{11} ], [ (x-1)^{11} ] \bigr)$.  Similarly, $\mathcal{C}_2$ is generated by
    $\bigl( [ (x-1)^{3} ], [ v_1(x)(x-1)^{3} ] \bigr)$ and $\bigl( [ v_2(x)(x-1)^{4} ], [ (x-1)^{4} ] \bigr)$. By Theorem~\ref{qsc1}, we obtain a quantum synchronizable code with parameters $(a_l, a_r)-[[50 + a_l + a_r, 10]]_5$, where $a_l + a_r < 25$. 
\end{example}

\subsection{Quantum synchronizable codes from QC codes with $m=lp^t$}
In this case, we set \( m = lp^t \), where \( \gcd(l, p) = 1 \). We begin by reviewing some basic concepts related to cyclotomic cosets \cite{Huffman2003}.  

For any integer \( s \), the \( q \)-cyclotomic coset modulo \( l \) containing \( s \) is defined as  
$C_{(s,l)} = \{s, sq, sq^2, \dots, sq^{i_s-1}\}$,
where \( i_s \) is the smallest positive integer such that \( sq^{i_s} \equiv s \mod l \). The smallest element in \( C_{(s,l)} \) is usually taken as its representative, and the set of all such representatives is denoted by \( T_l \). Let \( \alpha \) be a primitive \( l \)-th root of unity. The minimal polynomial corresponding to \( \alpha^s \) is given by  
$M_s(x) = \prod_{i\in C_{(s,l)}}(x-\alpha^i)$.
It follows that  
\[
x^{lp^t}-1 = (x^l-1)^{p^t} = \prod_{s\in T_l}M_s(x)^{p^t}.
\]  
\begin{lemma}\label{lemma2}
    Let \( \mathcal{C} \) be the 2-generator QC code generated by  
    \[
    \bigl( [\prod_{s\in T_l}M_s(x)^{r_{1s}}], [v_1(x) \prod_{s\in T_l}M_s(x)^{r_{1s}}] \bigr)
    \enspace \text{and}\enspace
    \bigl( [v_2(x) \prod_{s\in T_l}M_s(x)^{r_{2s}}], [\prod_{s\in T_l}M_s(x)^{r_{2s}}] \bigr),
    \]
    where \( \gcd(v_1(x)v_2(x)-1, x^{lp^t}-1)=1 \), and for any \( s\in T_l \), we have  $0 < r_{1s} < r_{2s} < p^t - r_{2(-s)}$. Then the Euclidean dual code \( \mathcal{C}^{\perp_E} \) is generated by    
    \[
\begin{aligned}
    \bigl( [\prod_{s\in T_l}M_s(x)^{p^t-{r_{1(-s)}}}], 
    [-&\overline{v}_2(x) \prod_{s\in T_l}M_s(x)^{p^t-{r_{1(-s)}}}] \bigr)\enspace \text{and}\\
     &\bigl( [-\overline{v}_1(x) \prod_{s\in T_l}M_s(x)^{p^t-{r_{2(-s)}}}], 
    [\prod_{s\in T_l}M_s(x)^{p^t-{r_{2(-s)}}}] \bigr).
\end{aligned}
\]
    Moreover, \( \mathcal{C}^{\perp_E} \subset \mathcal{C} \).
\end{lemma}
\begin{proof}
    Denote  
    \[
    g_1(x) = \prod_{s\in T_l}M_s(x)^{r_{1s}}, \quad g_2(x) = \prod_{s\in T_l}M_s(x)^{r_{2s}}.
    \]  
    The reciprocal polynomial of \( M_s(x) \) satisfies  
    \begin{align*}
        M_s(x)^* &= x^{\deg(M_s(x))} M_s\left(\frac{1}{x}\right) = x^{\lvert C_{(s,l)} \rvert} \prod_{i\in C_{(s,l)}} \left(\frac{1}{x} - \alpha^i\right) \\
        &= \prod_{i\in C_{(s,l)}} (-\alpha^i)(x - \alpha^{-i}) = \omega_s \prod_{i\in C_{(-s,l)}} (x - \alpha^i) = \omega_s M_{-s}(x),
    \end{align*}
    where \( \omega_s \in \mathbb{F}_q \). Thus, we obtain  
    \begin{align*}
        g_{1}^\perp(x) &=(\prod_{s\in T_l}M_{s}(x)^{p^t-r_{1s}})^* =\omega_{g_1} \prod_{s\in T_l}M_{-s}(x)^{p^t-r_{1s}}
        = \omega_{g_1} \prod_{s\in T_l}M_s(x)^{p^t-{r_{1(-s)}}}, \\
        g_{2}^\perp(x) &=(\prod_{s\in T_l}M_{s}(x)^{p^t-r_{2s}})^*= \omega_{g_2} \prod_{s\in T_l}M_{-s}(x)^{p^t-r_{2s}}
        = \omega_{g_2} \prod_{s\in T_l}M_s(x)^{p^t-{r_{2(-s)}}},
    \end{align*}
    where \( \omega_{g_1}, \omega_{g_2} \in \mathbb{F}_q \). According to Proposition~\ref{Edual}, the Euclidean dual code \( \mathcal{C}^{\perp_E} \) of \( \mathcal{C} \) is generated by 
    $\bigl( \prod_{s\in T_l}M_s(x)^{p^t-{r_{1(-s)}}}, 
    -\overline{v}_2(x) \prod_{s\in T_l}M_s(x)^{p^t-{r_{1(-s)}}} \bigr)$
     and $\bigl( -\overline{v}_1(x) \prod_{s\in T_l}M_s(x)^{p^t-{r_{2(-s)}}}, 
    \prod_{s\in T_l}M_s(x)^{p^t-{r_{2(-s)}}} \bigr)$. Furthermore, as a result of the constraint \( 0 < r_{1s} < r_{2s} < p^t - r_{2(-s)} \), the polynomials \( g_1(x), g_2(x), v_1(x), v_2(x) \) satisfy the divisibility conditions stated in Theorem~\ref{Edual-containing}, implying that \( \mathcal{C}^{\perp_E} \subset \mathcal{C} \). This completes the proof.
    
\end{proof}
\begin{theorem}\label{qsc2}
    Let \( \mathcal{C}_1 \) be the 2-generator  \([2lp^{t},k_1,d_1]\) QC code generated by  
    \[
    \bigl( [\prod_{s\in T_l}M_s(x)^{r_{1s}}], [v_1(x) \prod_{s\in T_l}M_s(x)^{r_{1s}}] \bigr)\enspace \text{and}\enspace
    \bigl( [v_2(x) \prod_{s\in T_l}M_s(x)^{r_{2s}}], [\prod_{s\in T_l}M_s(x)^{r_{2s}}] \bigr),
    \]
    where \( \gcd(v_1(x)v_2(x)-1, x^{lp^t}-1)=1 \) and for any \( s\in T_l \), we have $0 < r_{1s} < r_{2s} < p^t - r_{2(-s)}$. Let \( \mathcal{C}_2 \) be the 2-generator \([2lp^t,k_2,d_2]\) QC code generated by  
    \[
    \bigl( [\prod_{s\in T_l}M_s(x)^{j_{1s}}], [v_1(x) \prod_{s\in T_l}M_s(x)^{j_{1s}}] \bigr)\enspace \text{and}\enspace
    \bigl( [v_2(x) \prod_{s\in T_l}M_s(x)^{j_{2s}}], [\prod_{s\in T_l}M_s(x)^{j_{2s}}] \bigr),
    \]
    where \( 0<j_{1s}<r_{1s} \), \( 0<j_{1s}<j_{2s}<r_{2s} \). If there exists an integer \( s^{\prime}\in T_l \) with \( \gcd(s^{\prime},l)=1 \) satisfying either  
    \[
    {r_{1s^\prime}}-{j_{1s^\prime}}>p^{t-1},  
    \quad
    \text{or}  
    \quad
    {r_{1s^\prime}}-{j_{1s^\prime}}>0 \quad \text{and} \quad {r_{1s^{\prime\prime}}}-{j_{1s^{\prime\prime}}}>p^{t-1}  
    \]
    for some \( s^{\prime\prime} \neq s^{\prime} \in T_l \), then for any non-negative integers \( a_l, a_r \) satisfying \( a_l+a_r<lp^{t} \), there exists a quantum synchronizable code with parameters  
    \[
    (a_l,a_r)-[[2lp^t+a_l+a_r,2lp^t-2\sum_{s\in T_l}(r_{1s}+r_{2s})\cdot\lvert C_{(s,l)}\rvert]]_q,
    \]
    which corrects at least up to $\lfloor\frac{d_1 - 1}{2}\rfloor$ phase errors and at least   up to $\lfloor\frac{d_2 - 1}{2}\rfloor$ bit errors.
\end{theorem}
\begin{proof}
    Since \( 0<r_{1s}<r_{2s}< p^t-r_{2(-s)} \), it follows from Lemma~\ref{lemma2} that \( \mathcal{C}_1^{\perp_E}\subset \mathcal{C}_1 \).  
    Denote  
    \[
    f_1(x)=\prod_{s\in T_l}M_s(x)^{{r_{1s}}-{j_{1s}}},  
    \quad  
    f_2(x)=\prod_{s\in T_l}M_s(x)^{{r_{2s}}-{j_{2s}}}.
    \]
    Then, 
    \[
\begin{aligned}
   & \begin{pmatrix}
        [\prod_{s\in T_l}M_s(x)^{r_{1s}}] & [v_1(x)\prod_{s\in T_l}M_s(x)^{r_{1s}}]\\
        [v_2(x)\prod_{s\in T_l}M_s(x)^{r_{2s}}] & [\prod_{s\in T_l}M_s(x)^{r_{2s}}]
    \end{pmatrix} \\
    &=
    \begin{pmatrix}
        [f_1(x)] & \mathbf{0}\\
        \mathbf{0} & [f_2(x)]
    \end{pmatrix}
    \begin{pmatrix}
        [\prod_{s\in T_l}M_s(x)^{j_{1s}}] & [v_1(x)\prod_{s\in T_l}M_s(x)^{j_{1s}}]\\
        [v_2(x)\prod_{s\in T_l}M_s(x)^{j_{2s}}] & [\prod_{s\in T_l}M_s(x)^{j_{2s}}]
    \end{pmatrix}.
\end{aligned}
\]
    Therefore, we obtain that \( \mathcal{C}_1\subset \mathcal{C}_2 \).  
    In fact, it follows that 
$$\eta(x)=\gcd\bigl(\prod_{s\in T_l}M_s(x)^{r_{1s}}, v_2(x)\prod_{s\in T_l}M_s(x)^{r_{2s}}\bigr) = \prod_{s\in T_l}M_s(x)^{r_{1s}}.$$
    Moreover,  
    \[
    \prod_{s\in T_l}M_s(x)^{r_{1s}} = \prod_{s\in T_l}M_s(x)^{r_{1s}-j_{1s}} \prod_{s\in T_l}M_s(x)^{j_{1s}},
    \]
    the dimensions of \( \mathcal{C}_1 \) and \( \mathcal{C}_2 \) are given by  
    \[
    k_1 = 2lp^t - \sum_{s\in T_l} (r_{1s}+r_{2s}) \cdot \mid C_{s,l}\mid,
    k_2 = 2lp^t - \sum_{s\in T_l} (j_{1s}+j_{2s}) \cdot\mid C_{s,l}\mid.
    \]
    Since for any \( i\in T_l \) with \( \gcd(i,l)=1 \), we have $\text{ord}(M_i(x))=\text{ord}(\alpha^i)=\frac{l}{\gcd(i,l)}=l$, it follows that the order of \( \prod_{s\in T_l}M_s(x)^{r_{1s}-j_{1s}} \) is \( lp^t \) due to the given conditions on \( r_{1s} \) and \( j_{1s}\).  
    Finally, by Lemma~\ref{qc construction}, there exists a quantum synchronizable code with parameters  
    \[
    (a_l,a_r)-[[2lp^t+a_l+a_r,2lp^t-2\sum_{s\in T_l}(r_{1s}+r_{2s})\cdot\lvert C_{(s,l)}\rvert]]_q,
    \]
    where \( a_l+a_r<lp^t \), which completes the proof.
\end{proof}
\begin{example}
    Let $q = p = 3$, $t = 3$, $l = 3$, and define $v_1(x) = x^3 + x^2 + 2x + 1$, $v_2(x) = x^5 + 2x^4 + x^3 + 2x^2 + 2x + 2$. The code $\mathcal{C}_1$ is generated by $\bigl( [ (x - 1)^5 (x^2 + x + 1)^{12} ], [ v_1(x)(x - 1)^5 (x^2 + x + 1)^{12} ] \bigr)$ and $\bigl( [v_2(x)(x - 1)^{11} (x^2 + x + 1)^{13} ], [ (x - 1)^{11} (x^2 + x + 1)^{13} ] \bigr)$, while $\mathcal{C}_2$ is generated by $\bigl( [ (x - 1)^2 (x^2 + x + 1)^{2} ], [ v_1(x)(x - 1)^2 (x^2 + x + 1)^{2} ] \bigr)$ and $\bigl( [ v_2(x)(x - 1)^{7} (x^2 + x + 1)^{10} ], [ (x - 1)^{7} (x^2 + x + 1)^{10} ] \bigr)$. By Theorem~\ref{qsc2}, we obtain a quantum synchronizable code with parameters $(a_l, a_r)-[[162 + a_l + a_r, 30]]_3$, where $a_l + a_r < 81$. 
\end{example}

\begin{example}
    Let $q = p = 7$, $t = 2$, $l = 5$, and define $v_1(x) = x^3$, $v_2(x) = x^5 + x^2 + 5x + 4$. The code $\mathcal{C}_1$ is generated by $\bigl( [ (x - 1)^9 (x^4 + x^3 + x^2 + x + 1
)^{13} ], [ v_1(x)(x - 1)^{9} (x^4 + x^3 + x^2 + x + 1)^{13} ] \bigr)$ and $\bigl( [ v_2(x)(x - 1)^{14} (x^4 + x^3 + x^2 + x + 1)^{16} ], [ (x - 1)^{14} (x^4 + x^3 + x^2 + x + 1)^{16} ] \bigr)$, while $\mathcal{C}_2$ is generated by $\bigl( [ (x - 1)^7 (x^4 + x^3 + x^2 + x + 1)^{3} ], [ v_1(x)(x - 1)^7 (x^4 + x^3 + x^2 + x + 1)^{3} ] \bigr)$ and $\bigl( [ v_2(x)(x - 1)^{11} (x^4 + x^3 + x^2 + x + 1)^{12} ], [ (x - 1)^{11} (x^4 + x^3 + x^2 + x + 1)^{12} ] \bigr)$. By Theorem~\ref{qsc2}, we obtain a quantum synchronizable code with parameters $(a_l, a_r)-[[490 + a_l + a_r, 444]]_7$, where $a_l + a_r < 245$. 
\end{example}

\section{Conclusion}\label{sec7}

In this paper, we first derived the necessary and sufficient conditions for the self-orthogonality of QC codes under Euclidean, symplectic, and Hermitian inner products. Building on this, we applied these conditions to a special class of 2-generator QC codes to characterize their self-orthogonal and dual-containing properties. 
 Furthermore, we applied these results to construct quantum stabilizer codes and quantum synchronizable codes with good parameters. In future work, we plan to further investigate the conditions for quasi-twisted (QT) codes to be self-orthogonal via their generating sets. Moreover, we aim to apply the construction method proposed by Ezerman et al. in~\cite{Grassl2025} to obtain quantum error-correcting codes with larger minimum distance.

\bmhead{Acknowledgements}
We would like to express our sincere gratitude to Professor Markus Grassl for his valuable comments and constructive suggestions, which have greatly contributed to improving the quality of this work.

\section*{Declarations}


\begin{itemize}
\item \textbf{Funding}
This work was supported by Shandong Provincial Natural Science Foundation of China (ZR2023LLZ013, ZR2022MA061), Fundamental Research Funds for the Central Universities (23CX03003A), National Natural Science Foundation of China (61902429),  and China Education Innovation Research Fund (2022BL027).
\item \textbf{Conflict of interest} 
The authors declare no conflict of interest.
\item \textbf{Ethics approval and consent to participate}
Not applicable.
\item \textbf{Consent for publication}
Not applicable.
\item \textbf{Data availability} 
Not applicable.
\item \textbf{Materials availability}
Not applicable.
\item \textbf{Code availability} 
Not applicable.
\item \textbf{Author contribution}
Authors Mengying Gao and Yuhua Sun were mainly responsible for the theoretical analysis and experimental design, while authors Tongjiang Yan and Chun’e
Zhao contributed primarily to manuscript writing and table preparation.

\end{itemize}

\noindent

\end{document}